\documentclass[11pt, reqno]{amsart}
\usepackage{amsfonts}
\usepackage{amssymb}
\usepackage{latexsym}
\usepackage[final]{hyperref}
\usepackage[]{enumerate}
\usepackage{verbatim}
\usepackage[english]{babel} 
\usepackage{blindtext}
\usepackage[nohyperlinks,nolist]{acronym}
\usepackage{listings}

\usepackage[table]{xcolor}%
\newcommand{\major}{\cellcolor{blue!100}}
\newcommand{\minor}{\cellcolor{blue!25}}

\usepackage{graphicx}       
\usepackage[utf8]{inputenc} 
\usepackage{amsmath}

\usepackage{tikz}
\usepackage{cmap}
\usepackage{ulem}
\usepackage{soul}
\usepackage{comment}
\usepackage{fancybox}
\usepackage{relsize}
\usepackage{pifont}
\usepackage{subcaption}

\newcommand{\Abs}[1]{\left\vert #1 \right\vert}

\newcommand{\norm}[1]{\left\Vert #1 \right\Vert}

\newcommand{\R}{\mathbb{R}}

\newcommand{\etal}{\textit{et al.}}

\DeclareMathOperator{\im}{Im}
\DeclareMathOperator{\re}{Re}

\DeclareMathOperator{\erf}{erf}

\DeclareMathOperator{\Dom}{Dom}

\newtheorem{lemma}{Lemma}

\theoremstyle{definition}

\theoremstyle{remark}

\begin{acronym}
\acro{AO}{atomic orbital}
\acro{API}{Application Programmer Interface}
\acro{AUS}{Advanced User Support}
\acro{BEM}{Boundary Element Method}
\acro{BO}{Born-Oppenheimer}  
\acro{CBS}{complete basis set}
\acro{CC}{Coupled Cluster}
\acro{CTCC}{Centre for Theoretical and Computational Chemistry}
\acro{CoE}{Centre of Excellence}
\acro{DC}{dielectric continuum}  
\acro{DCHF}{Dirac-Coulomb Hartree-Fock}  
\acro{DFT}{density functional theory}  
\acro{DKH}{Douglas-Kroll-Hess}
\acro{EFP}{effective fragment potential}
\acro{ECP}{effective core potential}
\acro{EU}{European Union}
\acro{FFT}{Fast Fourier Transform}
\acro{GGA}{generalized gradient approximation}
\acro{GPE}{Generalized Poisson Equation}
\acro{GTO}{Gaussian Type Orbital}
\acro{HF}{Hartree-Fock}  
\acro{HPC}{high-performance computing}
\acro{Hylleraas}[HC]{Hylleraas Centre for Quantum Molecular Sciences}
\acro{IEF}{Integral Equation Formalism}
\acro{IGLO}{individual gauge for localized orbitals}
\acro{KB}{kinetic balance}
\acro{KS}{Kohn-Sham}
\acro{LAO}{London atomic orbital}
\acro{LAPW}{linearized augmented plane wave}
\acro{LDA}{local density approximation}
\acro{MAD}{mean absolute deviation}
\acro{maxAD}{maximum absolute deviation}
\acro{MM}{molecular mechanics}  
\acro{MCSCF}{multiconfiguration self consistent field}
\acro{MPA}{multiphoton absorption}
\acro{MRA}{multiresolution analysis}
\acro{MSDD}{Minnesota Solvent Descriptor Database}
\acro{MW}{multiwavelet}
\acro{NAO}{numerical atomic orbital}
\acro{NeIC}{nordic e-infrastructure collaboration}
\acro{KAIN}{Krylov-accelerated inexact Newton}
\acro{NMR}{nuclear magnetic resonance}
\acro{NP}{nanoparticle}  
\acro{NS}{non-standard}  
\acro{OLED}{organic light emitting diode}
\acro{PAW}{projector augmented wave}
\acro{PBC}{Periodic Boundary Condition}
\acro{PCM}{polarizable continuum model}
\acro{PW}{plane wave}
\acro{QC}{quantum chemistry}  
\acro{QM/MM}{quantum mechanics/molecular mechanics}  
\acro{QM}{quantum mechanics}  
\acro{RCN}{Research Council of Norway}
\acro{RMSD}{root mean square deviation}
\acro{RKB}{restricted kinetic balance}
\acro{SC}{semiconductor}
\acro{SCF}{self-consistent field}
\acro{STSM}{short-term scientific mission}
\acro{SAPT}{symmetry-adapted perturbation theory}
\acro{SERS}{surface-enhanced raman scattering}
\acro{WPREL}[WP1]{Work Package 1}
\acro{WPROP}[WP2]{Work Package 2}
\acro{WPAPP}[WP3]{Work Package 3}
\acro{WP}{Work Package}  
\acro{X2C}{exact two-component}
\acro{ZORA}{zero-order relativistic approximation}
\acro{ae}{almost everywhere}
\acro{BVP}{boundary value problem}
\acro{PDE}{partial differential equation}
\acro{RDM}{1-body reduced density matrix}
\acro{SCRF}{self-consistent reaction field}
\acro{IEFPCM}{Integral Equation Formalism \ac{PCM}}
\acro{FMM}{fast multipole method}
\acro{DD}{domain decomposition}
\acro{TRS}{time-reversal symmetry}
\acro{SI}{Supporting Information}
\acro{DHF}{Dirac--Hartree--Fock}
\acro{MO}{molecular orbital}
\end{acronym}

\title
[
	Multiresolution of free-particle propagator 
]
{
    Multiresolution of the one dimensional free-particle propagator. Part 1: Construction
}

\author[Dinvay]{Evgueni Dinvay}
\author[Zabelina]{Yuliya Zabelina}
\author[Frediani]{Luca Frediani}

\email{ evgueni.dinvay@uit.no }
\email{ jl.zabelina@gmail.com }
\email{ luca.frediani@uit.no }

\address
{
    Department of Chemistry
    \\
    UiT The Arctic University of Norway
    \\
    PO Box 6050 Langnes
    \\
    N-9037 Tromsø
    \\
    Norway
}

\date{\today}

\begin{document}

\begin{abstract}
The free-particle propagator,
a key operator in various algorithms for simulating the time evolution of the Schrödinger equation,
is studied.
A multiscale approximation of this propagator is constructed,
representing the semigroup associated with the free-particle Schrödinger operator in a multiwavelet basis.
This representation involves integrals of highly oscillatory functions.
These integrals are efficiently discretized using a contour deformation technique,
which addresses the challenges posed by earlier discretization methods. 
\end{abstract}

\keywords{
}
\maketitle
\section{Introduction}
\setcounter{equation}{0}


Let us consider the one-dimensional time-dependent Schrödinger equation
\begin{equation}
\label{general_schrodinger}
    i \partial_t u = - \partial_x^2 u + V(x, t) u
    ,
\end{equation}
complemented by the initial condition given at $t = 0$ by
\begin{equation}
\label{general_initial_data}
    u(0) = u_0 \in L^2(\mathbb R)
    .
\end{equation}
The space variable $x \in \mathbb R$ and the time variable $t > 0$.
In general a real-valued potential $V(x, t)$ can depend on the solution,
as for example, in water waves and Hartree-Fock theories.

The main objective of this work is to provide an effective
numerical discretization of the exponential operator
\(
    \exp \left( i t \partial_x^2 \right)
    ,
\)
where $t \in \mathbb R$ is a small parameter, not necessarily positive.
In what follows we assume that $t \ne 0$.
Our goal is to construct a numerical representation of this operator
to solve Equation \eqref{general_schrodinger} adaptively.
It is known that a use of multiwavelet bases \cite{Alpert1993}
may lead to adaptive numerical solutions for some
partial differential equations \cite{Alpert_Beylkin_Gines_Vozovoi2002}.
However, the time-evolution operator
\(
    \exp \left( i t \partial_x^2 \right)
\)
falls outside this operator class,
that are known to give sparse representations in this type of bases
\cite{Beylkin_Keiser1997, Fann_Beylkin_Harrison_Jordan2004}.
Nevertheless, it turns out that
\(
    \exp \left( i t \partial_x^2 \right)
\)
exhibits a unique peculiarity that can be leveraged in adaptive calculations,
as we demonstrate below.

We recall that the $t$-dependent exponential operator under consideration here
forms a semigroup
representing solutions of the free particle equation
\begin{equation}\label{free_particle_schrodinger}
    i \partial_t u + \partial_x^2 u = 0
    .
\end{equation}
In other words, for any square integrable $u_0 \in L^2(\R)$,
the function
\(
    u
    =
    \exp \left( i t \partial_x^2 \right)
    u_0
\)
solves Equation~\ref{free_particle_schrodinger}.
In numerical simulations, the parameter $t$ usually plays the role of a time (sub)step.
The importance of an effective discretization of this exponential operator,
at least for a time-independent potential $V$,
can be exemplified as follows.
Let $A, B$ be time independent linear operators.
One of the simplest second-order splitting methods can be expressed as follows
\begin{equation}
\label{leap_frog_integrator}
    e^{ At + Bt } = e^{ Bt/2 } e^{ At } e^{ Bt/2 } + \mathcal O \left( t^3 \right)
    .
\end{equation}
Taking $A = i \partial_x^2$ and $B = -i V$,
one can evolve Equation \eqref{general_schrodinger}
by applying the multiplication operator $e^{ Bt/2 }$ twice
and the pseudo-differential operator $e^{ At }$ once at each time step
starting from the initial function \eqref{general_initial_data}.
Such operator splitting simplifies the problem
\eqref{general_schrodinger}, \eqref{general_initial_data} significantly:
the exponential of the multiplicative operator is straightforward and the problem reduces to
finding an efficient representation for the exponential of the kinetic energy.
Below we will also make use of higher order schemes than \eqref{leap_frog_integrator}.
They will permit larger time steps,
which turns out to be especially crucial for multiwavelet representation
of the semigroup
\(
    \exp \left( i t \partial_x^2 \right)
    ,
\)
due to some limitations on how small a time step $t$ might be taken.

We are motivated by the desire of simulating attosecond electron dynamics.
In recent years there has been a growing interest in studying the electronic structure
of molecules under strong and ultra-fast laser pulses
\cite{Nisoli_Decleva_Calegari_Palacios_Martin}.
This is clearly testified by the 2018 and 2023 Nobel prizes in Physics\cite{Mourou_nobel_lecture2018, Strickland_nobel_lecture2018}. This could potentially lead to reaction control by attosecond laser pulses in the future.
From a computational perspective, electron dynamics in atoms and molecules is challenging:
standard atomic orbital approaches routinely used in Quantum Chemistry cannot be
applied in strong laser pulses, as simple numerical experiments
suggest that ionization plays a significant role and cannot be neglected \cite{Coccia_Luppi2021}.
In other words, the coupling with continuum or unbound states complicates the use of atomic orbitals
for the problem discretization.
Inevitably it demands the use of grid based methods, such as finite elements, for example.

In the case of hydrogen atom subjected to a laser field we have
the following three-dimensional Schrödinger equation
\begin{equation}
\label{hydrogen_schrodinger}
	i \partial_t \psi
	=
	- \frac 12 \Delta \psi
	- \frac 1{|x|} \psi
	+ \mathcal E(x, t) \psi
	,
\end{equation}
where the second term $- |x|^{-1}$ stands for the Coulomb potential
and $\mathcal E(x, t)$ stands for the external electric field generated by a laser.
It is normally complemented by the ground state
\(
	\psi_0(x) = e^{-|x|} / \sqrt{\pi}
\)
serving as an initial condition $\psi(x, 0) = \psi_0(x)$.

It can be advantageous to make use of adaptive methods to describe
the singular Coulomb interaction between the electron and the nucleus, as well as
the cusp in the eigenfunctions which follows from it.
This, however, sets some limitations on the method development.
Additionally, for many-body systems (atoms and molecules),
the potential becomes nonlinear due to the presence of electron-electron interactions.
The corresponding nonlinearities are normally presented in equations as singular integral convolutions. 
One such method is constituted by a \ac{MW} representation,
within the framework of multiresolution analysis \cite{Alpert_Beylkin_Gines_Vozovoi2002}.
This approach has demonstrated excellent results in achieving fast algorithms and high precision
for static quantum chemistry problems
\cite{Harrison_Fann_Yanai_Gan_Beylkin2004, Jensen_Saha_elephant2017, Frediani_Fossgaard_Fla_Ruud}.
However, discussing these static results is beyond the scope of the current contribution.
Furthermore, the research field is rapidly evolving,
and we recommend referring to a comprehensive review paper for a broader perspective \cite{Bischoff2019}.
An overview of the most recent advancements will also be available in an upcoming paper
\cite{Tantardini_Dinvay_Pitteloud_Gerez_Jensen_review2024}.

We aim to extend the multiwavelet approach to address time-dependent scenarios.
The first step in integrating the generic time-dependent Schrödinger equation,
similar to \eqref{hydrogen_schrodinger},
involves discretizing the semigroup $e^{i t \Delta}$.
Due to the limitations highlighted above,
the fast Fourier transform (FFT) cannot be relied upon for quantum chemistry applications.
This is why we find the multiscale approach promising.
To our knowledge, the only work on multiwavelet representation of the dynamical Shr\"odinger type
equations was conducted by Vence~\etal~\cite{Vence_Harrison_Krstic2012}.
They view the semigroup as a convolution
and mitigate the errors arising from its oscillatory behavior
by damping high-frequency oscillations in Fourier space.
The resulting kernel approximation can be seamlessly integrated into
the multiwavelet framework developed for static problems.

In contrast, we adopt a different approach
by leveraging the smoothing property of the semigroup instead of damping it.
It enables us to achieve a highly accurate multiscale approximation
of the free-particle propagator.
Moreover, it turns out that
this representation can be regarded as sparse,
as long as particular restrictions are met during numerical simulations.
Since we do not treat the semigroup as a convolution suitable for chemistry applications,
unlike the mentioned above method employed by Vence~\etal,
further development is necessary
to fully exploit our representation for studying electron dynamics in molecules.
Therefore, this article focuses exclusively on the one-dimensional case.
In the future, we plan to extend our results and simulate the three-dimensional
quantum mechanical problems.
It is here difficult to predict if a multiresolution representation of oscillatory operators such as the free particle propagator considered here
is going to be practical in 3D, due to its computational cost.
However, we demonstrate below a significant improvement of the results by Vence~\etal, at least in 1D.
Therefore, we believe that our approach will show an efficiency
comparable to other grid methods for small molecules in strong laser fields.

The paper is organised as follows.
In Section~\ref{Multiwavelet_bases_section} we introduce some preliminary notions and
recall the theory behind multiwavelet bases.
In Section~\ref{Review_on_time_evolution_operators_in_multiwavelet_bases}
we lay out the problem with some insights into the difficulties that we aim to overcome.
In Section~\ref{Contour_deformation_technique_section}, the multiresolution representation of
\(
    \exp \left( i t \partial_x^2 \right)
\)
is reduced to the evaluation of a special kind of integrals.
In Section~\ref{Contour_deformation_technique_section}
these integrals are evaluated by deforming the integration contour in a way that allows to exploit the smoothing property of this exponential operator.
Furthermore, we demonstrate that by carefully selecting the time parameter $t$
and polynomial approximation order,
its multiresolution representation will display an unusual sparsity pattern.
More specifically, the matrices used in calculations have
non-zero elements away from their main diagonals.
This sparsity depends significantly on the value of $t$,
necessitating the use of high-order numerical schemes for accurate wave propagation simulations.
To address this requirement, we introduce high-order symplectic integrators
in the second part of this work \cite{Dinvay2024}.

\section{Multiwavelet bases}\label{Multiwavelet_bases_section}
\setcounter{equation}{0}

This section provides a concise theoretical background on multiwavelet bases.
We are interested in approximations of functions in $L^2(\mathbb R)$.
However, for practical reasons, we have to restrict ourselves to
a sufficiently large interval.
Equation \eqref{general_schrodinger} can always be normalised in such a way
that the unit computational domain $[0, 1]$ will suit our needs,
while keeping the constant in front of the second derivative to be $-1$.
Note that any function $u \in L^2(0, 1)$ can be trivially extended
to a function in $L^2(\mathbb R)$, by setting it to zero outside $(0, 1)$.
On the contrary, any function from $L^2(\mathbb R)$ can be projected into
$L^2(0, 1)$, by multiplying it with the characteristic function of
the unit interval $(0, 1)$.
This remark allows us to make sense of the inclusion
\(
    L^2(0, 1) \subset L^2(\mathbb R)
    .
\)
Now for any
\(
    \mathfrak k \in \mathbb N
    = \{ 1, 2, \ldots\}
\)
and any
\(
    n \in \mathbb N_0
    = \{ 0, 1, 2, \ldots\}
\)
we introduce a space of piecewise polynomials $V_n^{\mathfrak k}$ as follows.
A function $f \in V_n^{\mathfrak k}$ provided that on each dyadic interval
\(
    (l / 2^n, (l + 1) / 2^n)
\)
with
\(
    l = 0, 1, \ldots, 2^n - 1
\)
it is a polynomial of degree less than $\mathfrak k$
and it is zero elsewhere.
Note that the space $V_n^{\mathfrak k}$ has dimension $2^n \mathfrak k$ and
\[
    V_0^{\mathfrak k}
    \subset
    V_1^{\mathfrak k}
    \subset
    \ldots
    \subset
    V_n^{\mathfrak k}
    \subset
    \ldots
    \subset
    L^2(0, 1) \subset L^2(\mathbb R)
    ,
\]
which defines a multiresolution analysis (MRA) \cite{Mallat_1989b, Mallat_1989a},
since the union
\(
    \bigcup_{n = 0}^{\infty} V_n^{\mathfrak k}
\)
is dense in $L^2(0, 1)$.
We refer to the given fixed number $\mathfrak k$ as the order of MRA and to
the integer variable $n$ as a scaling level.

Let $\phi_0, \ldots, \phi_{\mathfrak k - 1}$ be a basis of $V_0^{\mathfrak k}$
that will be called scaling functions.
The space $V_n^{\mathfrak k}$ is spanned
by $2^n \mathfrak k$ functions which are obtained from the scaling functions by
dilation and translation
\begin{equation}
\label{scaling_dilation_translation}
    \phi_{jl}^n(x)
    =
    2^{n/2}
    \phi_j( 2^n x - l)
    , \quad j = 0, \ldots, \mathfrak k - 1
    , \quad l = 0, \ldots, 2^n - 1
    .
\end{equation}
There is some freedom in choosing a basis in $V_0^{\mathfrak k}$,
at least for $\mathfrak k > 1$.
The first example appeared in \cite{Alpert1993}, and it consists of the Legendre scaling functions
\begin{equation}
\label{Legendre_scaling_functions}
    \phi_j(x)
    =
    \sqrt{ 2j + 1 } P_j(2x - 1)
    , \quad
    x \in (0, 1)
    , \quad
    j = 0, \ldots, \mathfrak k - 1
    ,
\end{equation}
where $P_j$ are standard Legendre polynomials.
Outside the unit interval $(0, 1)$ they are set to zero.
An alternative basis, proved to be more efficient
that is especially crucial in multi-dimensional numerical calculations,
was presented in \cite{Alpert_Beylkin_Gines_Vozovoi2002}.
It consists of the interpolating scaling functions
\begin{equation}
\label{interpolating_scaling_functions}
    \varphi_j(x)
    =
    \sqrt{ w_j }
    \sum_{m = 0}^{\mathfrak k - 1}
    \phi_m(x_j) \phi_m(x)
    , \quad
    j = 0, \ldots, \mathfrak k - 1
    ,
\end{equation}
where $\phi_m$ are the Legendre scaling functions \eqref{Legendre_scaling_functions}.
Here $x_0, \ldots, x_{\mathfrak k - 1}$ denote the roots of $P_{\mathfrak k}(2x - 1)$
and the quadrature weights
\(
    w_j
    =
    1 /
    (
        \mathfrak k P_{\mathfrak k}'(2x_j - 1) P_{\mathfrak k - 1}(2x_j - 1)
    )
    .
\)
The expansion coefficients of a general function $f$ in a scaling basis
are the integrals
\begin{equation}
\label{scaling_coefficients}
    s_{jl}^n(f)
    =
    \int_{\mathbb R}
    f(x) \phi_{jl}^n(x)
    dx
    =
    2^{-n/2}
    \int_0^1
    f( 2^{-n} (x + l) ) \phi_j(x)
    dx
    ,
\end{equation}
that can be evaluated numerically using
the Gauss-Legendre quadrature
\begin{equation}
\label{scaling_coefficients_quadrature}
    s_{jl}^n(f)
    \approx
    2^{-n/2}
    \sum_{m = 0}^{\mathfrak k - 1}
    w_m
    f( 2^{-n} (x_m + l) ) \phi_j(x_m)
    , \quad j = 0, \ldots, \mathfrak k - 1
    , \quad l = 0, \ldots, 2^n - 1
    .
\end{equation}
Note that the interpolating scaling functions \eqref{interpolating_scaling_functions}
satisfy
\(
    \varphi_j(x_m)
    =
    w_j^{-1/2} \delta_{jm}
\)
and so \eqref{scaling_coefficients_quadrature} simplifies to
\begin{equation}
\label{interpolating_scaling_coefficients_quadrature}
    s_{jl}^n(f)
    \approx
    2^{-n/2}
    \sqrt{w_j}
    f( 2^{-n} (x_m + l) )
\end{equation}
for this particular choice of a basis.

The multiwavelet space $W_n^{\mathfrak k}$ is defined as
the orthogonal complement of $V_n^{\mathfrak k}$ in $V_{n + 1}^{\mathfrak k}$,
so we have
\[
    V_{n + 1}^{\mathfrak k}
    =
    V_n^{\mathfrak k}
    \oplus
    W_n^{\mathfrak k}
    =
    V_0^{\mathfrak k}
    \oplus
    W_0^{\mathfrak k}
    \oplus
    W_1^{\mathfrak k}
    \oplus
    \ldots
    \oplus
    W_n^{\mathfrak k}
    .
\]
As above it is enough to define a basis
$\psi_0, \ldots, \psi_{\mathfrak k - 1}$ in $W_0^{\mathfrak k}$
and then transform it in line with the general rule \eqref{scaling_dilation_translation},
in order to get an orthonormal basis in each $W_n^{\mathfrak k}$.
The construction of a specific multiwavelet basis
that satisfies certain given restrictions is more involved.
We are using the one that was constructed in \cite{Alpert1993},
since it provides us with an additional vanishing moment property,
namely,
\begin{equation}
\label{wavelet_vanishing_moment}
    \int_0^1 \psi_j(x) x^m dx = 0
    , \quad
    m = 0, 1, \ldots, j + \mathfrak k - 1
    .
\end{equation}
The multiwavelet expansion coefficients of a general function $f$
are the integrals
\begin{equation}
\label{wavelet_coefficients}
    d_{jl}^n(f)
    =
    \int_{\mathbb R}
    f(x) \psi_{jl}^n(x)
    dx
    ,
\end{equation}
that could in principal be also evaluated with the help of
the Gauss-Legendre quadrature.
However, there are precision issues connected to this approach.
The multiwavelet coefficients \eqref{wavelet_coefficients}
are instead obtained from the higher level scaling coefficients $s_{jl}^{n + 1}(f)$,
where the latter are calculated by the quadrature rule
\eqref{scaling_coefficients_quadrature}.
This is possible thanks to the so called forward wavelet transform \cite{Cohen_1992}
\begin{equation}
\label{forward_wavelet_transform}
    \begin{pmatrix}
        \phi_l^n
        \\
        \psi_l^n
    \end{pmatrix}
    =
    U
    \begin{pmatrix}
        \phi_{2l}^{n + 1}
        \\
        \phi_{2l + 1}^{n + 1}
    \end{pmatrix}
\end{equation}
where the two bases of $V_{n + 1}^{\mathfrak k}$
are grouped in accordance with the short notation agreement
\[
    \phi_l^n
    =
    \begin{pmatrix}
        \phi_{0,l}^n (x)
        \\
        \ldots
        \\
        \phi_{\mathfrak k - 1,l}^n (x)
    \end{pmatrix}
    , \quad
    \psi_l^n
    =
    \begin{pmatrix}
        \psi_{0,l}^n (x)
        \\
        \ldots
        \\
        \psi_{\mathfrak k - 1,l}^n (x)
    \end{pmatrix}
\]
and $U$ stands for the unitary matrix
\begin{equation}
\label{unitary_filter_transformation}
    U
    =
    \begin{pmatrix}
        H^{(0)}
        &
        H^{(1)}
        \\
        G^{(0)}
        &
        G^{(1)}
    \end{pmatrix}
\end{equation}
consisting of $\mathfrak k \times \mathfrak k$-size filter blocks
\(
    H^{(0)}
    =
    \left(
        h_{pj}^{(0)}
    \right)
    , \ldots,
    G^{(1)}
    =
    \left(
        g_{pj}^{(1)}
    \right)
\)
depending only on the type of scaling basis (Legendre or interpolating) in use
\cite{Alpert_Beylkin_Gines_Vozovoi2002}.
Integrating \eqref{forward_wavelet_transform}
together with $f$, we immediately obtain the following relation
\begin{equation}
\label{decomposition_step}
    \begin{pmatrix}
        s_l^n
        \\
        d_l^n
    \end{pmatrix}
    =
    U
    \begin{pmatrix}
        s_{2l}^{n + 1}
        \\
        s_{2l + 1}^{n + 1}
    \end{pmatrix}
\end{equation}
with vectors $s_l^n$, $d_l^n$ consisting of coefficients 
$s_{jl}^n$, $d_{jl}^n$, $j = 0, \ldots, \mathfrak k - 1$, accordingly.

There are two main advantages with multiwavelets, compared to other discretization methods: (1)
the disjont support of the basis functions combined with the vanishing moments of the wavelets enables sparse and adaptive function representations, with a rigorous precision control; (2) some pseudo-differential operators  have sparse multiwavelet representations \cite{Alpert1993}, leading to fast (ideally lienarly scaling) algorithms.

Let $P^n, Q^n$ be the orthogonal projectors of $L^2(\mathbb R)$
into $V_n^{\mathfrak k}, W_n^{\mathfrak k}$, respectively.
Clearly,
\(
    Q^n = P^{n + 1} - P^n
    .
\)
Moreover, $P^n$ converges strongly to the projection
\(
    P : L^2(\mathbb R) \to L^2(0, 1)
    ,
\)
whereas $Q^n$ tends to zero.
Then any bounded operator $\mathcal T$ can be represented on the computational
domain $[0, 1]$ as
\begin{equation}
\label{nonstandard_form_expansion}
    P \mathcal T P
    =
    \lim_{n \to \infty} \mathcal T_n
    =
    \mathcal T_0
    +
    \sum_{n = 0}^{\infty}
    \left(
        \mathcal A_n + \mathcal B_n + \mathcal C_n
    \right)
    ,
\end{equation}
where we introduced the multiresolution restrictions
\begin{equation}
\label{nonstandard_form_restrictions}
    \mathcal A_n
    =
    Q^n \mathcal T Q^n
    , \quad
    \mathcal B_n
    =
    Q^n \mathcal T P^n
    , \quad
    \mathcal C_n
    =
    P^n \mathcal T Q^n
    , \quad
    \mathcal T_n
    =
    P^n \mathcal T P^n
    .
\end{equation}
In other words, each approximation $\mathcal T_n$ of $\mathcal T$ on $[0, 1]$
is defined by its restriction $\mathcal T_0$ in $V_0^{\mathfrak k}$
and by the collection of triplets
\(
    ( \mathcal A_0, \mathcal B_0, \mathcal C_0 )
    ,
    \ldots
    ,
    ( \mathcal A_n, \mathcal B_n, \mathcal C_n )
    .
\)
We refer to this representation as the nonstandard form and write shortly
\begin{equation}
\label{nonstandard_form}
    \mathcal T_n
    \cong
    \{
        \mathcal T_0
        ,
        ( \mathcal A_0, \mathcal B_0, \mathcal C_0 )
        ,
        \ldots
        ,
        ( \mathcal A_n, \mathcal B_n, \mathcal C_n )
    \}
    .
\end{equation}
This form is very useful for some classes of operators.
In particular,
Calder\'on-Zygmund operators have sparse representations for the matrices 
associated with $\mathcal A_n, \mathcal B_n, \mathcal C_n$,
leading to fast effective algorithms \cite{Beylkin_Coifman_Rokhlin}.
Another advantage of the nonstandard form is
the absence of coupling between scales when the operator is applied
(such a coupling results from a post-processing step
which relies only on the fast multiwavelet transform
\cite{Beylkin_Coifman_Rokhlin, Frediani_Fossgaard_Fla_Ruud}).
This is a key-feature as one wants to preserve the adaptivity of functions while the operator is applied.

Before proceeding to concrete examples of operators $\mathcal T$,
we introduce the following matrices, defined in the 
\(
    V_n^{\mathfrak k}
    \oplus
    W_n^{\mathfrak k}
    \) 
basis, as
\begin{equation}
\label{nonstandard_form_matrices}
\begin{aligned}
    \left[ \sigma_{l'l}^n \right]_{j'j}
    =
    \int_0^1
    \phi_{j'l'}^n(x) \mathcal T \phi_{jl}^n(x)
    dx
    , \qquad
    &
    \left[ \gamma_{l'l}^n \right]_{j'j}
    =
    \int_0^1
    \phi_{j'l'}^n(x) \mathcal T \psi_{jl}^n(x)
    dx
    ,
    \\
    \left[ \beta_{l'l}^n \right]_{j'j}
    =
    \int_0^1
    \psi_{j'l'}^n(x) \mathcal T \phi_{jl}^n(x)
    dx
    , \qquad
    &
    \left[ \alpha_{l'l}^n \right]_{j'j}
    =
    \int_0^1
    \psi_{j'l'}^n(x) \mathcal T \psi_{jl}^n(x)
    dx
    .
\end{aligned}
\end{equation}
For every scale $n \in \mathbb N_0$,
symbol $\sigma^n$ stands for a table of $2^n \times 2^n$-size,
where each element is itself a matrix of $\mathfrak k \times \mathfrak k$-size.
Similarly, the blocks
\(
    \alpha^n, \beta^n
\)
and $\gamma^n$,
associated with the triple 
\(
    ( \mathcal A_n, \mathcal B_n, \mathcal C_n )
    ,
\)
have the same forms.
Making use of the wavelet transform \eqref{forward_wavelet_transform}
one can easily deduce the following relation
\begin{equation}
\label{nonstandard_form_decomposition_step}
    \begin{pmatrix}
        \sigma_{l'l}^n
        &
        \gamma_{l'l}^n
        \\
        \beta_{l'l}^n
        &
        \alpha_{l'l}^n
    \end{pmatrix}
    =
    U
    \begin{pmatrix}
        \sigma_{2l', 2l}^{n + 1}
        &
        \sigma_{2l', 2l + 1}^{n + 1}
        \\
        \sigma_{2l' + 1, 2l}^{n + 1}
        &
        \sigma_{2l' + 1, 2l + 1}^{n + 1}
    \end{pmatrix}
    U^T
    .
\end{equation}

Finally,
we can describe a practical algorithm for the operator application
in the nonstandard form.
Let us for a given arbitrary $f \in L^2(0, 1)$ consider
the sequences
\(
    \widetilde d_l^n
    ,
    \widetilde s_l^n
    ,
    \widehat s_l^n
\)
with
\(
    n \in \mathbb N_0
    ,
    l = 0, 1, \ldots, 2^n - 1
\)
of elements from $\mathbb C^{\mathfrak k}$ defined by the following iterative procedure
\[
    \widetilde d_{l'}^n
    =
    \sum_{l = 0}^{2^n - 1}
    \left(
        \alpha_{l'l}^n
        d_l^n(f)
        +
        \beta_{l'l}^n
        s_l^n(f)
    \right)
    , \quad
    \widetilde s_{l'}^n
    =
    \sum_{l = 0}^{2^n - 1}
    \gamma_{l'l}^n
    d_l^n(f)
\]
and
\begin{equation*}
    \begin{aligned}
        \widehat s_{2m}^n
        =
        H^{(0)}
        \left(
            \widehat s_m^{n - 1}
            +
            \widetilde s_m^{n - 1}
        \right)
        +
        G^{(0)}
        \widetilde d_m^{n - 1}
        \\
        \widehat s_{2m + 1}^n
        =
        H^{(1)}
        \left(
            \widehat s_m^{n - 1}
            +
            \widetilde s_m^{n - 1}
        \right)
        +
        G^{(1)}
        \widetilde d_m^{n - 1}
    \end{aligned}
    , \quad
    n \geqslant 1
    , \quad
    \text{where }
    \widehat s_0^0
    =
    \sigma_{00}^0
    s_0^0(f)
    .
\end{equation*}

\begin{lemma}
    Let $f \in L^2(0, 1)$, $\mathfrak n \in \mathbb N_0$
    and $g = \mathcal T_{\mathfrak n + 1} f$.
    Then its coefficients can be obtained by the above iterative procedure
    as
    \begin{equation}
    \label{operator_application}
        d_l^{\mathfrak n}(g)
        =
        \widetilde d_l^{\mathfrak n}
        , \quad
        s_l^{\mathfrak n}(g)
        =
        \widetilde s_l^{\mathfrak n}
        +
        \widehat s_l^{\mathfrak n}
    \end{equation}
    at the finest scale $\mathfrak n$.
    The rest coefficients, at scales $n = 0, \ldots, \mathfrak n - 1$,
    are found by \eqref{decomposition_step}.
\end{lemma}

\begin{proof}
It is enough to prove \eqref{operator_application},
since the scale relation \eqref{decomposition_step} is already known.
Moreover,
the first equality in \eqref{operator_application} is obvious.
Now we can notice that the second equality in \eqref{operator_application}
clearly holds for the finest scale $\mathfrak n = 0$.
Thus the statement follows by induction
over $\mathfrak n$
from the equalities
\begin{equation*}
\begin{aligned}
    \int_0^1
    \phi_{pl}^n(x) \phi_{j, 2m}^{n + 1}(x)
    dx
    =
    h_{pj}^{(0)} \delta_{lm}
    , \qquad
    &
    \int_0^1
    \psi_{pl}^n(x) \phi_{j, 2m}^{n + 1}(x)
    dx
    =
    g_{pj}^{(0)} \delta_{lm}
    \\
    \int_0^1
    \phi_{pl}^n(x) \phi_{j, 2m + 1}^{n + 1}(x)
    dx
    =
    h_{pj}^{(1)} \delta_{lm}
    , \qquad
    &
    \int_0^1
    \psi_{pl}^n(x) \phi_{j, 2m + 1}^{n + 1}(x)
    dx
    =
    g_{pj}^{(1)} \delta_{lm}
\end{aligned}
\end{equation*}
that are straightforward to check.
Further details are omitted.
\end{proof}

\section{Review on time evolution operators in multiwavelet bases}
\label{Review_on_time_evolution_operators_in_multiwavelet_bases}
\setcounter{equation}{0}

There are two main examples of the operator $\mathcal T$ that we regard here, namely,
the heat
\(
    \exp \left( t \partial_x^2 \right)
\)
and Schrödinger
\(
    \exp \left( i t \partial_x^2 \right)
\)
evolution operators.
The former one is probably the most used and studied operator in relation
to application of multiwavelets.
In this section we review how the heat propagation operator
\(
    \exp \left( t \partial_x^2 \right)
\)
is discretised in a multiwavelet basis,
and which problems arise when the same techniques are directly adapted to the Schrödinger semigroup
\(
    \exp \left( i t \partial_x^2 \right)
    .
\)

\subsection{Finite interval representation with Dirichlet boundary conditions}
\label{Finite_interval_representation_with_Dirichlet_boundary_conditions}

The Laplacian $\partial_x^2$ together with the Direchlet boundary conditions,
or more precisely, defined on the domain
\[
    \Dom \left( \partial_x^2 \right)
    =
    \left\{
        u \in H_2^2(0, 1)
        \, \big| \,
        u(0) = u(1) = 0
    \right\}
    ,
\]
constitute an unbounded self-adjoint operator
in $L^2(0, 1)$.

The collection of functions
\(
    \left\{
        x \mapsto
        \sqrt 2 \sin( \pi \nu x )
    \right\}
    _{ \nu \in \mathbb N }
\)
forms an orthonormal basis in $L^2(0, 1)$.
Let $\mathbb F$ be the Fourier transform associated with this basis,
that is
\[
    \mathbb F u (\nu)
    =
    \int_0^1
    u(x)
    \sqrt 2 \sin( \pi \nu x )
    dx
    .
\]
The elements of this basis are eigenfunctions of the Laplacian $\partial_x^2$, with eigenvalues $- (\pi \nu)^2$, $\nu = 1, 2, \ldots$.
Therefore, according to the general spectral theory
for any function $\Psi$ we can define the following operator
\[
    \Psi \left( \partial_x^2 \right)
    u(x)
    =
    \sum_{\nu = 1}^{\infty}
    \Psi \left( - (\pi \nu)^2 \right)
    \int_0^1
    u(x')
    \sqrt 2 \sin( \pi \nu x' )
    dx'
    \sqrt 2 \sin( \pi \nu x )
\]
on functions $u \in L^2(0, 1)$.
Alternatively, this can be written as a Fourier transform, followed by multiplication by $\Psi \left( - (\pi \nu)^2 \right)$ and inverse Fourier transform:
\[
    \Psi \left( \partial_x^2 \right)
    =
    \mathbb F^{-1}
    \Psi \left( - (\pi \nu)^2 \right)
    \mathbb F
    .
\]
%

The Fourier transform $\mathbb F$ is an isomorphism and its inverse
is defined as
\[
    \mathbb F^{-1}
    :
    \ell^2( \mathbb N ) \ni \alpha
    \mapsto
    \sum_{\nu = 1}^{\infty}
    \alpha_{\nu} \sqrt 2 \sin( \pi \nu x )
    \in L^2(0, 1)
    .
\]
The main example we are interested in is
$\Psi(\lambda) = \exp(it\lambda)$.
The heat exponent with
$\Psi(\lambda) = \exp(t\lambda)$
has been previously considered in \cite{Alpert_Beylkin_Gines_Vozovoi2002}.
The propagator
\(
    \exp \left( i t \partial_x^2 \right)
\)
is defined through the symbol
\(
    \exp \left( -i t (\pi \nu)^2 \right)
\)
as described above. We would like
to represent this operator in a multiwavelet basis
in the same spirit as in \cite{Alpert_Beylkin_Gines_Vozovoi2002}.
Following the same logic we start with the projection of $u(x)$ in a \ac{MW} basis:
\[
    P^n u(x)
    =
    \sum_{l = 0}^{N - 1}
    \sum_{j = 0}^{\mathfrak k - 1}
    s_{jl}^n \phi_{jl}^n (x)
    , \quad
    N = 2^n
    ,
\]
and then consider its time evolution under the action of the operator 
\(
    \mathcal T = \Psi \left( \partial_x^2 \right)
    .
\)
The function
\(
    \mathcal T u
\)
is approximated by
\(
    \mathcal T_n u
\)
which
lies in the span of piecewise Legendre polynomials.
We have
\[
    \mathcal T_n u
    =
    \sum_{l = 0}^{N - 1}
    \sum_{j = 0}^{\mathfrak k - 1}
    \widetilde s_{jl}^n \phi_{jl}^n
    ,
\]
where the coefficients $\widetilde s_{jl}^n$ are obtained by applying $\Psi \left( \partial_x^2 \right)$ onto $P^n u(x)$, and then taking the inner product with the scaling function $\phi_{jl}^n (x)$,
\[
    \widetilde s_{jl}^n
    =
    \int_0^1
    \Psi \left( \partial_x^2 \right)
    P^n u(x)
    \phi_{jl}^n (x) dx
    =
    \int_0^1
    \mathbb F^{-1}
    \left(
        \Psi \left( - (\pi \nu)^2 \right)
        \mathbb F P^n u(\nu)
    \right)
    (x)
    \phi_{jl}^n (x) dx
    .
\]
The rightmost integral is the inner product in $L^2(0, 1)$.
It can be rewritten in terms of the inner product
in $\ell^2(\mathbb N)$ by noticing that
\(
    \mathbb F^* = \mathbb F^{-1}
\)
which follows immediately from the Parseval's identity.
Thus
\[
    \widetilde s_{j'l'}^n
    =
    \sum_{\nu = 1}^{\infty}
    \Psi \left( - (\pi \nu)^2 \right)
    \mathbb F P^n u(\nu)
    \mathbb F
    \phi_{j'l'}^n
    (\nu)
    =
    \sum_{l = 0}^{N - 1}
    \sum_{j = 0}^{\mathfrak k - 1}
    s_{jl}^n
    \sum_{\nu = 1}^{\infty}
    \Psi \left( - (\pi \nu)^2 \right)
    \mathbb F
    \phi_{jl}^n
    (\nu)
    \mathbb F
    \phi_{j'l'}^n
    (\nu)
    .
\]
Here we make a first remark on the convergence of this series.
Our main example
\(
    \Psi \left( - (\pi \nu)^2 \right)
    =
    \exp \left( -i t (\pi \nu)^2 \right)
\)
is a function bounded with respect to $\nu \in \mathbb N$.
Both
\(
    \mathbb F
    \phi_{jl}^n
\)
and
\(
    \mathbb F
    \phi_{j'l'}^n
\)
are $\ell^2$-sequences.
Therefore the series is absolutely convergent
by the Cauchy-Schwarz's inequality.
So far we do not have much information about the rate of convergence,
though one can anticipate a significant cancellation effect at high
frequencies $\nu$ due to oscillations caused by the exponent.

Let us simplify further the expression for
\(
    \mathbb F
    \phi_{jl}^n
    (\nu)
\)
in the following way
\begin{multline*}
    \mathbb F
    \phi_{jl}^n
    (\nu)
    =
    \int_0^1
    \phi_{jl}^n
    (x)
    \sqrt 2 \sin( \pi \nu x )
    dx
    =
    \frac 1{ \sqrt N }
    \int_0^1
    \phi_j(x)
    \sqrt 2 \sin \left( \pi \nu \frac{x + l}N \right)
    dx
    \\
    =
    - \frac i{ \sqrt {2N} }
    \int_0^1
    \phi_j(x)
    \left(
        \exp \left( i \pi \nu \frac{x + l}N \right)
        -
        \exp \left( - i \pi \nu \frac{x + l}N \right)
    \right)
    dx
    .
\end{multline*}
If we introduce the usual Fourier transform as
\begin{equation}
\label{Fourier_transform}
    \widehat f (\xi)
    =
    \mathcal F f (\xi)
    =
    \int_\R f(x) e^{-i \xi x} dx
\end{equation}
then after simple algebraic manipulations we get
\begin{multline*}
    \mathbb F
    \phi_{jl}^n
    (\nu)
    \mathbb F
    \phi_{j'l'}^n
    (\nu)
    =
    \frac 1N \re
    \left(
        \mathcal F \phi_j \left( \frac{ \pi \nu }N \right)
        \mathcal F \phi_{j'} \left( -\frac{ \pi \nu }N \right)
        \exp \frac{ i \pi \nu (l' - l) }N
        \right.
        \\
        \left.
        -
        \mathcal F \phi_j \left( -\frac{ \pi \nu }N \right)
        \mathcal F \phi_{j'} \left( -\frac{ \pi \nu }N \right)
        \exp \frac{ i \pi \nu (l' + l) }N
    \right)
    .
\end{multline*}
Summing up we obtain
\[
    \widetilde s_{j'l'}^n
    =
    \sum_{l = 0}^{N - 1}
    \sum_{j = 0}^{\mathfrak k - 1}
    s_{jl}^n
    \left(
        \left[ \sigma_{l' - l}^{1n} \right]_{j'j}
        -
        \left[ \sigma_{l' + l}^{2n} \right]_{j'j}
    \right)
    , \quad
    N = 2^n
    ,
\]
which shows that the matrix elements of $\mathcal T_n$ can be written as
\(
    \sigma_{l'l}^{n}
    =
    \sigma_{l' - l}^{1n}
    -
    \sigma_{l' + l}^{2n}
\)
in
\(
    \{ \phi_{jl}^n \}
    ,
\)
where we introduced the following notations
\[
    \left[ \sigma_l^{1n} \right]_{j'j}
    =
    \sum_{\nu = 1}^{\infty}
    \Psi \left( - (\pi \nu)^2 \right)
    \frac 1N \re
    \left(
    \mathcal F \phi_j \left( \frac{ \pi \nu }N \right)
    \overline{
        \mathcal F \phi_{j'} \left( \frac{ \pi \nu }N \right)
    }
    \exp \frac{ i \pi \nu l }N
    \right)
    ,
\]
\[
    \left[ \sigma_l^{2n} \right]_{j'j}
    =
    \sum_{\nu = 1}^{\infty}
    \Psi \left( - (\pi \nu)^2 \right)
    \frac 1N \re
    \left(
    \overline{
        \mathcal F \phi_j \left( \frac{ \pi \nu }N \right)
    }
    \overline{
        \mathcal F \phi_{j'} \left( \frac{ \pi \nu }N \right)
    }
    \exp \frac{ i \pi \nu l }N
    \right)
    .
\]

The advantage of the above equations is that the Fourier
transforms above can be computed exactly, because they involve the calculation of integrals of polynomials multiplied by exponentials.
Moreover, the recursion formula for the Legendre polynomials
implies a recursion for the corresponding Fourier transforms.
Indeed, integrating the following relation
\cite{Abramowitz}
with an exponential weight
\[
    (2j + 1) P_j(x)
    =
    \partial_x
    ( P_{j+1}(x) - P_{j-1}(x) )
\]
one obtains
\[
    \frac{2j + 1}{i\xi} \int_{-1}^1 P_j(x) e^{-i \xi x} dx
    =
    \int_{-1}^1 P_{j+1}(x) e^{-i \xi x} dx
    -
    \int_{-1}^1 P_{j-1}(x) e^{-i \xi x} dx
    .
\]
Hence, the Fourier-transformed Legendre scaling functions \eqref{Legendre_scaling_functions}
display the following recursion formulas
\begin{equation}
\label{Fourier_Legendre_recursion}
    \mathcal F \phi_{j+1}(\xi)
    =
    \sqrt{ \frac{ 2j + 3 }{ 2j - 1 } }
    \mathcal F \phi_{ j - 1 }(\xi)
    +
    \frac{ 2 \sqrt{ (2j + 1)(2j + 3) } }{ i\xi }
    \mathcal F \phi_j(\xi)
    ,
\end{equation}
whereas the first two Fourier transforms
\begin{equation}
\label{Fourier_Legendre_0}
    \mathcal F \phi_0(\xi)
    =
    \frac 1{ i\xi }
    \left(
        1 - e^{ -i\xi }
    \right)
    ,
\end{equation}
\begin{equation}
\label{Fourier_Legendre_1}
    \mathcal F \phi_1(\xi)
    =
    \frac { i \sqrt 3 }{ \xi }
    \left(
        1 + e^{ -i\xi }
    \right)
    +
    \frac { 2 \sqrt 3 }{ \xi^2 }
    \left(
        e^{ -i\xi } - 1
    \right)
\end{equation}
can be obtained by
direct calculations.

We remark again here that each Fourier transform
\(
    \mathcal F \phi_j \left( \frac{ \pi \nu }N \right)
    =
    \mathcal O
    \left(
        \frac N{\nu}
    \right)
    ,
\)
when $\nu \to \infty$,
as follows from the calculation of Fourier transforms.
This is prohibitively slow for our main example of bounded $\Psi$.
Keeping only the first $N^m$ terms of the sums we get
\[
    \left[ \sigma_l^{1n} \right]_{j'j}
    =
    \sum_{\nu = 1}^{N^m}
    \Psi \left( - (\pi \nu)^2 \right)
    \frac 1N \re
    \mathcal F \phi_j \left( \frac{ \pi \nu }N \right)
    \overline{
        \mathcal F \phi_{j'} \left( \frac{ \pi \nu }N \right)
    }
    \exp \frac{ i \pi \nu l }N
    +
    \mathcal O
    \left(
        \frac 1{N^{m-1}}
    \right)
    ,
\]
and similarly for
\(
    \left[ \sigma_l^{2n} \right]_{j'j}
    .
\)
Indeed, the reminder is bounded up to a constant by
\[
    \frac 1N
    \sum_{\nu = N^m + 1}^{\infty}
    \frac 1{( \nu / N )^2}
    \leqslant
    \int_{N^{m-1}}^{\infty}
    \frac{dx}{x^2}
    =
    \frac 1{N^{m-1}}
    .
\]
This is problematic since in practice an acceptable precision
demands setting $m = 3$ or larger,
which corresponds to the tolerance
\(
    \mathcal O
    \left(
        \frac 1{N^2}
    \right)
    =
    \mathcal O
    \left(
        2^{-2n}
    \right)
    .
\)
This requires at least $N^3$ multiplications.
On the other hand, the reminder can be estimated more precisely
by making use of the high oscillation of $\Psi$ mentioned above.
The advantage of the smoothing property of the exponential operator
will be exploited in Section~\ref{Contour_deformation_technique_section}.

For the heat equation the series
\(
    \left[ \sigma_l^{1n} \right]_{j'j}
\)
and
\(
    \left[ \sigma_l^{2n} \right]_{j'j}
\)
converge significantly faster thanks to
the decay
\(
    \Psi \left( - (\pi \nu)^2 \right)
    =
    \exp \left( - t (\pi \nu)^2 \right)
    .
\)
Moreover,
the corresponding non-standard form blocks
\(
    \alpha^n, \beta^n
\)
and $\gamma^n$
defined in \eqref{nonstandard_form_matrices}
turn out to be effectively sparse.
We remark here that the entries of the matrices
\(
    \left[ \sigma_l^{1n} \right]_{j'j}
    ,
    \left[ \sigma_l^{2n} \right]_{j'j}
\)
are calculated only once for a fixed and given time step $t$, or, in case of higher-order temporal integration schemes, for a sequence of time substeps.
Then each of these matrices is applied to a vector
\(
    \{ s_{jl}^n \}
\)
following a fast numerical procedure developed in
\cite{Beylkin_Coifman_Rokhlin}.
When the non-standard form of the operator is used, the sparse approximations
of
\(
    \alpha^n, \beta^n
\)
and $\gamma^n$
are applied to 
\(
    \{ s_{jl}^n \}
    ,
    \{ d_{jl}^n \}
    .
\)

\subsection{Convolution representation on the real line}
\label{Convolution_representation_on_the_real_line}

On the real line, the operator
\(
    \mathcal T = \Psi \left( \partial_x^2 \right)
\)
is a convolution of the form
\begin{equation}
\label{convolution}
    \mathcal T u (x)
    =
    \int_{\mathbb R}
    K(x - y) u(y)
    dy
\end{equation}
with the kernel
\(
    K
    =
    \mathcal F_{\xi}^{-1}
    \Psi
    \left(
        - \xi^2
    \right)
    .
\)
Its matrix elements \eqref{nonstandard_form_matrices} are the integrals
\[
    \left[ \sigma_{l'l}^n \right]_{j'j}
    =
    \int_0^1
    \int_0^1
    K(x - y)
    \phi_{j'l'}^n(x) \phi_{jl}^n(y)
    dx
    dy
\]
simplifying to
\[
    \left[ \sigma_{l'l}^n \right]_{j'j}
    =
    \left[ \sigma_{l' - l}^n \right]_{j'j}
    =
    \frac 1N
    \int_{-1}^1
    K
    \left(
        \frac{ z + l' - l }N
    \right)
    \Phi_{j'j}(z)
    dz
\]
according to \eqref{scaling_dilation_translation},
where $N = 2^n$ and the so called correlation functions
\[
    \Phi_{j'j}(z)
    =
    \int_0^1
    \phi_{j'}(x)
    \phi_j(x - z)
    dx
\]
have been introduced.
Each of them is continuous with the support in $[-1, 1]$.
Their restrictions to either $[-1, 0]$ or $[0, 1]$ are
polynomials of order at most $2 \mathfrak k - 1$.
Therefore,
one may expand the correlation functions in
\(
    L^2(-1, 1)
\)
as
\[
    \Phi_{j'j}(z)
    =
    \sum_{p = 0}^{2 \mathfrak k - 1}
    \left(
        c_{j'jp}^{(+)} \phi_p(z)
        +
        c_{j'jp}^{(-)} \phi_p(z + 1)
    \right)
    ,
\]
where the cross correlation coefficients
\[
\left \{ \,
\begin{aligned}
    &
    c_{j'jp}^{(+)}
    =
    \int_0^1
    \int_0^1
    \phi_{j'}(x)
    \phi_j(x - z)
    \phi_p(z)
    dx dz
    \\
    &
    c_{j'jp}^{(-)}
    =
    \int_0^1
    \int_0^1
    \phi_{j'}(x)
    \phi_j(x - z + 1)
    \phi_p(z)
    dx dz
\end{aligned}
\right.
    , \quad
    j', j = 0, \ldots, \mathfrak k - 1
    , \quad
    p = 0, \ldots, 2 \mathfrak k - 1
    ,
\]
are easily tabulated.
Thus the convolution operator $\mathcal T$
in the orthonormal collection
\(
    \{ \phi_{jl}^n \}
\)
has the matrix elements
\(
    \left[ \sigma_{l'l}^n \right]_{j'j}
    =
    \left[ \sigma_{l' - l}^n \right]_{j'j}
\)
depending only on the distance
\(
    l' - l
    =
    1 - N, \ldots, N - 1
\)
to the diagonal.
These elements may be evaluated via
\begin{equation}
\label{cross_correlation_convolution_operator_representation}
    \left[ \sigma_l^n \right]_{j'j}
    =
    \frac 1{\sqrt N }
    \sum_{p = 0}^{2 \mathfrak k - 1}
    \left(
        c_{j'jp}^{(+)} s_{p, l}^n(K)
        +
        c_{j'jp}^{(-)} s_{p, l - 1}^n(K)
    \right)
    ,
\end{equation}
where $N = 2^n$ and the scaling coefficients of the kernel $K$
are defined by \eqref{scaling_coefficients}.

The numerical evaluation of convolution operators
is reviewed in \cite{Fann_Beylkin_Harrison_Jordan2004},
for instance.
Equation \eqref{cross_correlation_convolution_operator_representation}
gives the pure scaling component of the operator, the first integral in \eqref{nonstandard_form_matrices}.
The other components can be calculated by the decomposition transformation
\eqref{nonstandard_form_decomposition_step} that
simplifies to
\begin{equation}
\label{nonstandard_form_convolution_decomposition_step}
    \begin{pmatrix}
        \sigma_l^n
        &
        \gamma_l^n
        \\
        \beta_l^n
        &
        \alpha_l^n
    \end{pmatrix}
    =
    U
    \begin{pmatrix}
        \sigma_{2l}^{n + 1}
        &
        \sigma_{2l-1}^{n + 1}
        \\
        \sigma_{2l+1}^{n + 1}
        &
        \sigma_{2l}^{n + 1}
    \end{pmatrix}
    U^T
\end{equation}
for the convolution operator $\mathcal T$
and $l = - N + 1, - N + 2, \ldots, N - 1$
standing for the distance to the diagonal.

We can now demonstrate how
the multiwavelet framework can be used
for the exponential heat operator
\(
    \exp \left( t \partial_x^2 \right)
    .
\)
It can be regarded as a convolution operator in $L^2(\mathbb R)$
of the form
\[
    \exp \left( t \partial_x^2 \right)
    u(x)
    =
    \frac 1{ \sqrt{4 \pi t} }
    \int_\R
    \exp
    \left(
        - \frac{ (x - y)^2 }{4t}
    \right)
    u(y) dy
    , \quad
    t > 0
    ,
\]
as shown in Ref~\cite{Evans}.
It is associated with the heat equation
\(
    \partial_t u = \partial_x^2 u
\)
on the real line $\mathbb R$.
The Green's function scaling coefficients
\(
    s_{p, l}^n(K)
\)
can be calculated easily with high precision.
Then from \eqref{cross_correlation_convolution_operator_representation},
\eqref{nonstandard_form_convolution_decomposition_step}
one obtains the corresponding \ac{NS}-form matrices
\(
    \alpha^n, \beta^n, \gamma^n
\)
with $n = 0, 1, \ldots$.
Gaussian kernels are separable and possess a narrow, diagonally banded structure
in their \ac{NS}-form, which has already been exploited for Chemistry applications in the past, in particular by approximating the Poisson and Helmholtz kernel as a sum of Gaussians~\cite{Harrison_Fann_Yanai_Gan_Beylkin2004,Fann_Beylkin_Harrison_Jordan2004,Frediani_Fossgaard_Fla_Ruud}, similar to the one appearing in this heat semigroup convolution.

An attempt to adapt the above approach
to the time-dependent Schrödinger equation
\eqref{general_schrodinger}
was made in
\cite{Vence_Harrison_Krstic2012}.
The exponential operator
\(
    \exp \left( i t \partial_x^2 \right)
\)
can be regarded as a convolution operator in $L^2(\mathbb R)$
of the form
\[
    \exp \left( i t \partial_x^2 \right)
    u(x)
    =
    \frac{ \exp( -i \pi / 4 ) }{ \sqrt{4 \pi t} }
    \int_\R
    \exp
    \left(
        \frac{ i(x - y)^2 }{4t}
    \right)
    u(y) dy
    , \quad
    t > 0
    ,
\]
see \cite{Evans}.
The oscillatory behaviour of the kernel makes it not feasible to 
calculate the convolution with a given high precision.
Therefore, in
\cite{Vence_Harrison_Krstic2012}
this difficulty has been overcome by damping high frequencies:
the kernel was approximated by
\[
    K(x)
    \approx
    \mathcal F_{\xi}^{-1}
    \left(
        e^{-it \xi^2} \chi(\xi)
    \right)
    (x)
    ,
\]
where $\chi(\xi)$ is a smooth cut off function.
The problem with such an approach is that
it is difficult to control the error and
the cut off $\chi(\xi)$ should be tuned for each particular problem.
Moreover, as in
Subsection~\ref{Finite_interval_representation_with_Dirichlet_boundary_conditions} this approach ignores
the smoothing property of the exponential operator.

\section{Contour deformation technique}
\label{Contour_deformation_technique_section}
\setcounter{equation}{0}

Kaye~\etal{} \cite{Kaye_Barnett_Greengard2022} have exploited the smoothing property of
the exponential operator
\(
    \mathcal T
    =
    \exp \left( i t \partial_x^2 \right).
\)
They regarded Equation \eqref{general_schrodinger}
in the frequency domain.
Returning back to the physical domain,
where the potential $V(x, t)$ is applied,
they calculated the inverse Fourier transform $\mathcal F^{-1}$
over a specifically deformed contour $\Gamma$ instead of $\mathbb R$,
which allowed them to increase the precision significantly,
while keeping the advantage of using \ac{FFT} based schemes.
In this section we adopt their approach
in order to get a precise
multiresolution representation of $\mathcal T$.
As in Section
\ref{Finite_interval_representation_with_Dirichlet_boundary_conditions},
we can write
\[
    \mathcal T_n u
    =
    P^n
    \exp \left( i t \partial_x^2 \right)
    P^n
    u
    =
    P^n
    \mathcal F^{-1}
    e^{-it \xi^2}
    \mathcal F
    P^n
    u
    =
    \sum_{l = 0}^{N - 1}
    \sum_{j = 0}^{\mathfrak k - 1}
    \widetilde s_{jl}^n \phi_{jl}^n
    ,
\]
where now the operator $\mathcal T$ is unitary in $L^2(\mathbb R)$.
Similarly, we obtain
\[
    \widetilde s_{j'l'}^n
    =
    \sum_{l = 0}^{N - 1}
    \sum_{j = 0}^{\mathfrak k - 1}
    s_{jl}^n
    \left[ \sigma_{l' - l}^{n} \right]_{j'j}
    , \quad
    N = 2^n
    ,
\]
where the matrix
\(
    \left[ \sigma_{l' l}^{n} \right]_{j'j}
\)
defined in
\eqref{nonstandard_form_matrices}
depends only on the distance $l' - l$ to the diagonal,
as explained in Section
\ref{Convolution_representation_on_the_real_line}.
In terms of the introduced notations we have the following expression
\begin{multline*}
    \left[ \sigma_l^{n} \right]_{j'j}
    =
    \frac 1{2 \pi N}
    \int_{\R}
    \exp
    \left(
        \frac{ i \xi l }N - it \xi^2
    \right)
    \mathcal F \phi_j \left( \frac{ \xi }N \right)
    \overline{
        \mathcal F \phi_{j'} \left( \frac{ \xi }N \right)
    }
    d\xi
    \\
    =
    \frac 1{2 \pi}
    \int_{\R}
    \exp
    \left(
        i \xi l - it N^2 \xi^2
    \right)
    \mathcal F \phi_j (\xi)
    \mathcal F \phi_{j'} (-\xi)
    d\xi
    .
\end{multline*}
Note that
\(
    \left[ \sigma_l^{n} \right]_{j'j}
    =
    \left[ \sigma_{-l}^{n} \right]_{jj'}
    .
\)
The Fourier transforms of the scaling functions are given
by \eqref{Fourier_Legendre_recursion},
\eqref{Fourier_Legendre_0}, \eqref{Fourier_Legendre_1}.
Clearly,
the integrand can be extended to an entire function
of complex variable $\zeta \in \mathbb C$.
In two quadrants of $\mathbb C$-plane we have
$\re \left( - it N^2 \zeta^2 \right) < 0$,
which suggests that deforming the integration contour into them, would lead to a more effective calculation of the integral with respect to $\zeta$.
We could in principle deform the integration contour
exactly as was done in
\cite{Kaye_Barnett_Greengard2022},
so that
the new contour $\Gamma$ would be chosen depending on the sign of $t$
and an accuracy parameter $\varepsilon$.
Without loss of generality one can assume $t > 0$ from now on.
Their contour $\Gamma_H$ is kept
in an $H$-neighbourhood of the real axis $\mathbb R$,
i.e. in the band $| \im \zeta | \leqslant H$,
where the bound $H > 0$ is introduced
in order to avoid
multiplication of big and small numbers
while calculating the integral.
A direct repetition of their argument leads to the optimal bound
\begin{equation}
\label{H_bound}
    H
    =
    \frac 1N \log \frac{ \pi \varepsilon }{ 2 \varepsilon_{\text{mach}} }
    ,
\end{equation}
where $\varepsilon$ is the desired precision
and $\varepsilon_{\text{mach}}$ is the machine epsilon
(the precision limit of floating-point arithmetic on a computer,
approximately $2^{-52}$ for double type),
see the details in Subsection~\ref{Haar_multiresolution_analysis_subsection}.
However,
in our case we can use as an advantage the fact
that we are able to calculate
Fourier transforms of piecewise polynomials $\phi_j$
exactly.
Indeed, all the integrands are sums of exponents
(up to powers of $1 / \zeta$),
therefore, 
combining them together one can try to avoid inaccurate multiplications.
Without this restriction, we are able to take a contour
that will allow us to exploit the smoothing semigroup property
at its best.
Namely,
we choose a contour to be the line $\Gamma = (1 - i) \mathbb R$
oriented with angle $- \pi/ 4$ to the real line.
On such contour $\Gamma$ the main exponent part
\(
    - it N^2 \zeta^2
    =
    - t N^2 |\zeta|^2
    ,
\)
which guarantees fast integral convergence.

The nonstandard form matrices \eqref{nonstandard_form_matrices}
take their final form
\begin{equation}
\label{sigma_contour_matrix}
    \left[ \sigma_l^{n} \right]_{pj}
    =
    \frac 1{2 \pi}
    \int_{\Gamma}
    \exp
    \left(
        i \zeta l - it N^2 \zeta^2
    \right)
    \mathcal F \phi_j (\zeta)
    \mathcal F \phi_{p} (-\zeta)
    d\zeta
    ,
\end{equation}
\begin{equation}
\label{alpha_contour_matrix}
    \left[ \alpha_l^{n} \right]_{pj}
    =
    \frac 1{2 \pi}
    \int_{\Gamma}
    \exp
    \left(
        i \zeta l - it N^2 \zeta^2
    \right)
    \mathcal F \psi_j (\zeta)
    \mathcal F \psi_{p} (-\zeta)
    d\zeta
    ,
\end{equation}
\begin{equation}
\label{beta_contour_matrix}
    \left[ \beta_l^{n} \right]_{pj}
    =
    \frac 1{2 \pi}
    \int_{\Gamma}
    \exp
    \left(
        i \zeta l - it N^2 \zeta^2
    \right)
    \mathcal F \phi_j (\zeta)
    \mathcal F \psi_{p} (-\zeta)
    d\zeta
    ,
\end{equation}
\begin{equation}
\label{gamma_contour_matrix}
    \left[ \gamma_l^{n} \right]_{pj}
    =
    \frac 1{2 \pi}
    \int_{\Gamma}
    \exp
    \left(
        i \zeta l - it N^2 \zeta^2
    \right)
    \mathcal F \psi_j (\zeta)
    \mathcal F \phi_{p} (-\zeta)
    d\zeta
    .
\end{equation}
As we shall see below, the matrices are effectively sparse,
due to
\(
    \mathcal F \psi_j (0) = 0
\)
of at least order $\mathfrak k$.
We point out here that Fourier transform
\(
    \mathcal F \psi_j (\xi)
\)
can be easily found analytically and extended
to the complex plane $\zeta \in \mathbb C$ as
\[
    \widehat \psi_m (\zeta)
    =
    \frac 1{\sqrt 2}
    \sum_{j = 0}^{\mathfrak k - 1}
    \widehat \phi_j \left( \frac{\zeta}2 \right)
    \left[
        g_{ij}^{(0)}
        +
        g_{ij}^{(1)} \exp \left( - \frac {i\zeta}2 \right)
    \right]
    ,
\]
where $g_{ij}^{(0)}$ and $g_{ij}^{(1)}$ are elements
of the filter matrices $G^{(0)}$ and $G^{(1)}$, respectively.
The justification of the contour deformation is straightforward,
so we omit the proof.

Before we continue with the evaluation of these integrals, it is worth
to make a couple of remarks on the behaviour of the integrands in
\eqref{sigma_contour_matrix}-\eqref{gamma_contour_matrix}.
As can be seen from Equations
\eqref{Fourier_Legendre_recursion}-\eqref{Fourier_Legendre_1},
the Legendre Fourier transform extensions
$\widehat \phi_j(\zeta)$
have a removable singularity at zero $\zeta = 0$.
The same is true for
the interpolating Fourier transform extensions
$\widehat \varphi_j(\zeta)$
and the wavelet Fourier transform extensions
$\widehat \psi_j(\zeta)$.
This suggests that they can be approximated by power series
around zero.
Moreover, taking into account the Gaussian factor
\(
    \exp
    \left(
        - t N^2 |\zeta|^2
    \right)
    ,
\)
one may expect that these power series will provide good numerical values also away from the origin $\zeta = 0$.
Below, we will mostly focus on developing this idea further.
The last remark concerns the price to pay for the contour deformation.
As it will be obvious below, apart from the Gaussian fast decreasing factor we also get the increasing factor
\(
    \exp
    \left(
        (|l| + 1) |\zeta| / \sqrt 2
    \right)
    ,
\)
that turns out to be crucial, due to round up errors,
when the time step $t$ is small and the scale $n$ is coarse (small $n$).
This problem is in practice overcome by the algorithm used to construct the \ac{NS} form of the operator: the first integral \eqref{sigma_contour_matrix}
is computed with high precision at the finest scale $n$ required,
depending on $t$. The operator at coarser scales are obtained by virtue of the \ac{MW} transform \eqref{nonstandard_form_convolution_decomposition_step}.

\subsection{Legendre scaling functions}
\label{Legendre_scaling_functions_subsection}

Our goal is to calculate the integrals
\begin{equation}
\label{sigma_contour_matrix_a}
    \left[ \sigma_l^n \right]_{pj}
    (a)
    =
    \frac 1{2 \pi}
    \int_{\Gamma}
    \exp
    \left(
        i \zeta l - ia \zeta^2
    \right)
    \widehat \phi_j (\zeta)
    \widehat \phi_p (-\zeta)
    d\zeta
    , \quad
    0 \leqslant p, j < \mathfrak k
    ,
\end{equation}
where $a = t N^2 = t 4^n$.
We can represent the multiplication of
Fourier transforms as a $\zeta$-power series.
This will transform this integral into a series of
integrals $J_k(a, l)$, that can be evaluated exactly,
with some coefficients
depending solely on $p, j, k$.
These coefficients need to be computed only once and tabulated,
as they are problem-independent.
We will call them correlation coefficients.
We make an extensive use of \eqref{Fourier_Legendre_recursion}
which is valid for integers $j \geqslant 1$.

Firstly, we notice that $\widehat \phi_j(\zeta)$
has a root at zero of order $j$.
Indeed,
\[
    \partial_{\zeta}^m \widehat \phi_j(0)
    =
    (-i)^m \int_{\mathbb R} \phi_j(x) x^m dx
\]
that equals zero for any $0 \leqslant m < j$
and a non-zero provided $m = j$.
In particular,
\(
    \widehat \phi_0(0) = 1
\)
and
\(
    \widehat \phi_1(0) = 0
    .
\)
In other words, the Taylor series of the entire function
\(
    \widehat \phi_j(\zeta)
\)
starts with the power $j$.

Secondly, we can see from Equations
\eqref{Fourier_Legendre_recursion}-\eqref{Fourier_Legendre_1}
that each
\(
    \widehat \phi_j(\zeta)
\)
is a combination
of powers of $1 / \zeta$
and of exponentials
\(
    e^{\pm i \zeta}
    .
\)
More precisely,
\(
    \widehat \phi_j(\zeta)
    =
    \Phi_j(-i\zeta)
    ,
\)
where
\begin{equation*}
    \Phi_j(x)
    =
    A_0^j \frac 1x + \ldots + A_j^j \frac 1{x^{j+1}}
    +
    B_0^j \frac {e^x}x + \ldots + B_j^j \frac {e^x}{x^{j+1}}
    .
\end{equation*}

These coefficients are real and can be found exactly.
Indeed, from \eqref{Fourier_Legendre_0} we deduce
\begin{equation}
\label{AB_coefficients_for_j_is_0}
    A_0^0 = -1
    , \quad
    B_0^0 = 1
    ,
\end{equation}
and from \eqref{Fourier_Legendre_1} we deduce
\begin{equation}
\label{AB_coefficients_for_j_is_1}
    A_0^1 = \sqrt{3}
    , \quad
    A_1^1 = 2 \sqrt{3}
    , \quad
    B_0^1 = \sqrt{3}
    , \quad
    B_1^1 = -2 \sqrt{3}
    .
\end{equation}
For $j \geqslant 1$
from \eqref{Fourier_Legendre_recursion}
we deduce the following relation
\begin{equation}
\label{A_coefficients_recurrence}
\begin{aligned}
    A_0^{j+1}
    &=
    \sqrt{ \frac{ 2j + 3 }{ 2j - 1 } }
    A_0^{j-1}
    \\
    A_1^{j+1}
    &=
    \sqrt{ \frac{ 2j + 3 }{ 2j - 1 } }
    A_1^{j-1}
    -
    2 \sqrt{ (2j + 1)(2j + 3) }
    A_0^j
    \\
    &\ldots \qquad \ldots \qquad \ldots \qquad \ldots \qquad \ldots
    \\
    A_{j-1}^{j+1}
    &=
    \sqrt{ \frac{ 2j + 3 }{ 2j - 1 } }
    A_{j-1}^{j-1}
    -
    2 \sqrt{ (2j + 1)(2j + 3) }
    A_{j-2}^j
    \\
    A_j^{j+1}
    &=
    -
    2 \sqrt{ (2j + 1)(2j + 3) }
    A_{j-1}^j
    \\
    A_{j+1}^{j+1}
    &=
    -
    2 \sqrt{ (2j + 1)(2j + 3) }
    A_j^j
\end{aligned}
\end{equation}
which also holds for $B_m^j$ with $j \geqslant 1$.

It is now possible to write the Taylor series
about zero $\zeta = 0$ for the entire function
\(
    \widehat \phi_j (\zeta)
    \widehat \phi_p (-\zeta)
    =
    \Phi_j (-i\zeta)
    \Phi_p (i\zeta)
    .
\)
For $j, p \geqslant 0$ we have
\begin{equation*}
    \Phi_j (-x)
    \Phi_p (x)
    =
    \sum_{m=0}^j
    \sum_{q=0}^p
    \frac{ (-1)^{m+1} }{ x^{m+q+2} }
    \left(
        A_m^j A_q^p
        +
        B_m^j B_q^p
        +
        A_m^j B_q^p e^x
        +
        B_m^j A_q^p e^{-x}
    \right)
    .
\end{equation*}
Expanding the exponentials $e^{\pm x}$ in powers of $x$ we obtain
\begin{equation*}
    \Phi_j (-x)
    \Phi_p (x)
    =
    \sum_{m=0}^j
    \sum_{q=0}^p
    \sum_{k = j + p + 2 + m + q}^{\infty}
    \frac{ (-1)^{m+1} x^{k - m - q - 2} }{ k! }
    \left(
        A_m^j B_q^p
        +
        (-1)^k
        B_m^j A_q^p
    \right)
    ,
\end{equation*}
where we have discarded small powers,
since the Taylor series starts from the power $j + p$
in $x$-variable.
We can change the summation variable in the third sum as
\begin{equation*}
    \Phi_j (-x)
    \Phi_p (x)
    =
    \sum_{m=0}^j
    \sum_{q=0}^p
    \sum_{k = 0}^{\infty}
    \frac{ (-1)^{m+1} x^{k + j + p} }{ (k + 2 + j + p + m + q)! }
    \left(
        A_m^j B_q^p
        +
        (-1)^{k + j + p + m + q}
        B_m^j A_q^p
    \right)
    .
\end{equation*}
This series can be written down as
\begin{equation*}
    \Phi_j (-x)
    \Phi_p (x)
    =
    \sum_{k = 0}^{\infty}
    \frac{ C_{jp}^k x^{k + j + p} }{ (k + 2 + j + p)! }
    ,
\end{equation*}
where
\begin{equation*}
    C_{jp}^k
    =
    \sum_{m=0}^j
    \sum_{q=0}^p
    \frac{ (-1)^{m+1} (k + 2 + j + p)! }{ (k + 2 + j + p + m + q)! }
    \left(
        A_m^j B_q^p
        +
        (-1)^{k + j + p + m + q}
        B_m^j A_q^p
    \right)
    .
\end{equation*}
In particular,
\(
    C_{00}^k = 1 + (-1)^k
    ,
\)
that constitute all cross correlation coefficients required
for the Haar multiresolution.

There is an obvious problem here, namely,
\(
    \left| A_m^j \right|
    \sim
    (4j)^m
\)
which is numerically problematic, and the same is true for $B_m^j$.
So one may try to balance multiplication
by factorising it in the following way
\begin{equation*}
    C_{jp}^k
    =
    \sum_{m=0}^j
    \sum_{q=0}^p
    \frac{ (-1)^{m+1} (k + 2 + j + p)! (4j)^m (4p)^q }{ (k + 2 + j + p + m + q)! }
    \left(
        \widetilde A_m^j \widetilde B_q^p
        +
        (-1)^{k + j + p + m + q}
        \widetilde B_m^j \widetilde A_q^p
    \right)
    ,
\end{equation*}
with
\begin{equation*}
    \widetilde A_m^j
    =
    \frac
    {A_m^j}
    {(4j)^m}
    , \quad
    \widetilde B_q^p
    =
    \frac
    {B_q^p}
    {(4p)^q}
    , \quad
    \text{and so on.}
\end{equation*}
These new coefficients satisfy
the following relation
\begin{equation}
\label{tilde_B_coefficients_recurrence}
\begin{aligned}
    \widetilde B_0^{j+1}
    &=
    \sqrt{ \frac{ 2j + 3 }{ 2j - 1 } }
    \widetilde B_0^{j-1}
    \\
    \widetilde B_1^{j+1}
    &=
    \sqrt{ \frac{ 2j + 3 }{ 2j - 1 } }
    \frac{ j - 1 }{ j + 1 }
    \widetilde B_1^{j-1}
    -
    \sqrt{ \frac {(2j + 1)(2j + 3)}{(2j + 2)(2j + 2)} }
    \widetilde B_0^j
    \\
    \widetilde B_2^{j+1}
    &=
    \sqrt{ \frac{ 2j + 3 }{ 2j - 1 } }
    \left( \frac{ j - 1 }{ j + 1 } \right)^2
    \widetilde B_2^{j-1}
    -
    \sqrt{ \frac {(2j + 1)(2j + 3)}{(2j + 2)(2j + 2)} }
    \frac j{ j + 1 }
    \widetilde B_1^j
    \\
    &\ldots \qquad \ldots \qquad \ldots \qquad \ldots \qquad \ldots
    \\
    \widetilde B_{j-1}^{j+1}
    &=
    \sqrt{ \frac{ 2j + 3 }{ 2j - 1 } }
    \left( \frac{ j - 1 }{ j + 1 } \right)^{j-1}
    \widetilde B_{j-1}^{j-1}
    -
    \sqrt{ \frac {(2j + 1)(2j + 3)}{(2j + 2)(2j + 2)} }
    \left( \frac j{ j + 1 } \right)^{j-2}
    \widetilde B_{j-2}^j
    \\
    \widetilde B_j^{j+1}
    &=
    -
    \sqrt{ \frac {(2j + 1)(2j + 3)}{(2j + 2)(2j + 2)} }
    \left( \frac j{ j + 1 } \right)^{j-1}
    \widetilde B_{j-1}^j
    \\
    \widetilde B_{j+1}^{j+1}
    &=
    -
    \sqrt{ \frac {(2j + 1)(2j + 3)}{(2j + 2)(2j + 2)} }
    \left( \frac j{ j + 1 } \right)^j
    \widetilde B_j^j
\end{aligned}
\end{equation}
which also holds for $\widetilde A_m^j$ with $j \geqslant 1$, whereas
the initial coefficients are
\begin{equation*}
    \widetilde A_0^0 = -1
    , \quad
    \widetilde A_0^1 = \sqrt{3}
    , \quad
    \widetilde A_1^1 = \frac{\sqrt{3}}2
    ,
\end{equation*}
\begin{equation}
\label{tilde_first_B_coefficients}
    \widetilde B_0^0 = 1
    , \quad
    \widetilde B_0^1 = \sqrt{3}
    , \quad
    \widetilde B_1^1 = - \frac{\sqrt{3}}2
    .
\end{equation}

These formulas are well balanced and can be easily implemented. 
It turns out that due to multiple self-cancellations,
the correlation coefficients stay bounded.
For all the ranges of $k, j, p$ of interest
our simulations give $|C_{jp}^k| < 10$.
Moreover, we can further simplify the expression for
\(
    C_{jp}^k
    ,
\)
by exploiting some symmetries of $A$- and $B$-coefficients.
For example, we can guarantee that $C_{jp}^{2k + 1} = 0$
for all $k = 0, 1, \ldots$.
\begin{lemma}
    For any $j = 0, 1, \ldots$ and $m = 0, 1, \ldots, j$ it holds true
    that
    \(
        A_m^j = ( -1 )^{1 + j - m} B_m^j
        .
    \)
\end{lemma}
\begin{proof}
The proof is split in two steps.
Firstly, we claim the statement for $m = 0$, namely,
\begin{equation}
\label{coefficient_lemma_proof_step_1}
    A_0^j = ( -1 )^{1 + j} B_0^j
    , \quad
    j \in \mathbb N_0
    .
\end{equation}
Indeed,
for $j = 0$ it follows from \eqref{AB_coefficients_for_j_is_0}
and 
for $j = 1$ it follows from \eqref{AB_coefficients_for_j_is_1}.
For the other integers it follows by the induction
\[
    A_0^{j + 1}
    =
    c_1(0, j + 1)
    A_0^{j - 1}
    =
    c_1(0, j + 1)
    (-1)^j
    B_0^{j - 1}
    =
    (-1)^j
    B_0^{j + 1}
\]
due to the relation \eqref{A_coefficients_recurrence}
which both $A$- and $B$-coefficients satisfy to.
Similarly, we obtain
\begin{equation}
\label{coefficient_lemma_proof_step_1_m_not_zero}
    A_{j - 1}^j = B_{j - 1}^j
    , \quad
    A_j^j = - B_j^j
    , \quad
    j \in \mathbb N
    .
\end{equation}

The second step is to prove that
\begin{equation}
\label{coefficient_lemma_proof_step_2}
    A_z^{z + s} = ( -1 )^{1 + s} B_z^{z + s}
\end{equation}
for all
\(
    z, s \in \mathbb N_0
    .
\)
One could proceed by using induction over two variables,
but we can easily reduce the problem to the standard one variable induction.
Indeed, let us consider a standard bijection between
\(
    n \in \mathbb N
\)
and
\(
    (z, s) \in \mathbb N_0^2
    ,
\)
mapping the increment $n \mapsto n + 1$
either as
\begin{equation}
\label{bijection_increment_1}
    (z, s)
    \mapsto
    (z + 1, s - 1)
\end{equation}
or as
\begin{equation}
\label{bijection_increment_2}
    z \mapsto 0
    \mbox{ and }
    z + s
    \mapsto
    z + s + 1
    .
\end{equation}
Let $\mathcal P(n)$
stand for the statement \eqref{coefficient_lemma_proof_step_2}
with $z(n), s(n)$.
The induction base follows from the previous step
\eqref{coefficient_lemma_proof_step_1}.
Let $\mathcal P(1), \ldots, \mathcal P(n)$ hold true.
We need to check the validity of $\mathcal P(n + 1)$.
If the increment $n \mapsto n + 1$ corresponds to \eqref{bijection_increment_2},
then $s(n) = 0$, $s(n+1) = z(n) + 1$
and
$\mathcal P(n + 1)$
stands for the statement
\begin{equation*}
    A_0^{z(n) + s(n) + 1} = ( -1 )^{2 + z(n)} B_0^{z(n) + s(n) + 1}
\end{equation*}
holding true by
\eqref{coefficient_lemma_proof_step_1}.
It is left to check
$\mathcal P(n + 1)$
for the case of the correspondence $n \mapsto n + 1$
to \eqref{bijection_increment_1}, namely,
we need to prove
\begin{equation*}
    A_{z(n) + 1}^{z(n) + s(n)} = ( -1 )^{s(n)} B_{z(n) + 1}^{z(n) + s(n)}
    .
\end{equation*}
This is obviously true when $s(n) = 1$ or $s(n) = 2$
due to \eqref{coefficient_lemma_proof_step_1_m_not_zero}.
For other possible $z = z(n)$ and $s = s(n)$,
i.e. $s \geqslant 3$,
it follows from the recurrence relation \eqref{A_coefficients_recurrence}
as
\begin{multline*}
    A_{z + 1}^{z + s}
    =
    c_1( z + 1, z + s )
    A_{z + 1}^{z + s - 2}
    -
    c_2( z + 1, z + s )
    A_{z}^{z + s - 1}
    \\
    =
    c_1( z + 1, z + s )
    (-1)^{s - 2}
    B_{z + 1}^{z + s - 2}
    -
    c_2( z + 1, z + s )
    (-1)^s
    B_{z}^{z + s - 1}
    =
    ( -1 )^s B_{z + 1}^{z + s}
\end{multline*}
and the induction assumption on the validity of the first $n$ statements.

\end{proof}

From this lemma one can easily derive the final expression
\begin{equation}
\label{cross_correlation_coefficients}
    C_{jp}^k
    =
    (-1)^j \left( 1 + (-1)^k \right)
    \sum_{m=0}^j
    \sum_{q=0}^p
    \frac{ (k + 2 + j + p)! (4j)^m (4p)^q }{ (k + 2 + j + p + m + q)! }
    \widetilde B_m^j \widetilde B_q^p
    .
\end{equation}
In particular, $C_{jp}^k = 0$ for odd indices $k$.
Finally, we show that these coefficients are uniformly bounded
with respect to $k$.

\begin{lemma}
\label{cross_correlation_coefficients_bound_lema}
    For any non-negative integers $k, j, p$ the following bound holds true
    \begin{equation*}
        \left| C_{jp}^k \right|
        \leqslant
        2^{j + p + 1} \sqrt{(2j + 1)(2p + 1)}
        \left \{
        \begin{aligned}
            (j + 1)
            \frac{ (p/j)^{p + 1} - 1 }{ p/j - 1}
            , \quad
            \mbox{ provided }
            p = 3(j - 1)
            ,
            \\
            (p + 1)
            \frac{ (j/p)^{j + 1} - 1 }{ j/p - 1}
            , \quad
            \mbox{ provided }
            j = 3(p - 1)
            ,
            \\
            \frac
            {
                \left( (4j)^{j + 1} - (3 + j + p)^{j + 1} \right)
                \left( (4p)^{p + 1} - (3 + j + p)^{p + 1} \right)
            }
            {
                ( 3(j - 1) - p )
                ( 3(p - 1) - j )
                (3 + j + p)^{j + p}
            }
            \text{, other.}
        \end{aligned}
        \right.
    \end{equation*}
\end{lemma}

\begin{proof}

We start by showing that
\[
    \max_{m = 0, 1, \ldots, j}
    \left| \widetilde B_m^j \right|
    \leqslant
    2^j \sqrt{2j + 1}
\]
for each $j \in \mathbb N_0$.
Indeed, $\widetilde B_0^j$ can be calculated directly
by the first line in the recurrence relation \eqref{tilde_B_coefficients_recurrence}
with the first coefficients \eqref{tilde_first_B_coefficients}.
It yields
\[
    \widetilde B_0^j
    =
    \sqrt{2j + 1}
    , \quad
    j \in \mathbb N_0
    .
\]
The last two lines in \eqref{tilde_B_coefficients_recurrence} give
\[
    \left| \widetilde B_{j - 1}^j \right|
    \leqslant
    \left| \widetilde B_0^1 \right|
    , \quad
    \left| \widetilde B_j^j \right|
    \leqslant
    \left| \widetilde B_1^1 \right|
    , \quad
    j \in \mathbb N
    .
\]
In particular,
\[
    \left| \widetilde B_{j - 1}^j \right|
    ,
    \left| \widetilde B_j^j \right|
    \leqslant
    \sqrt{2j + 1}
    , \quad
    j \in \mathbb N
    .
\]
Now let us consider an array $M_z^j$ with
\(
    z = 0, \ldots, j
\)
and
\(
    j \in \mathbb N_0
\)
satisfying
\[
    M_z^{j + 1}
    =
    \left( \frac{ j - 1 }{ j + 1 } \right)^z
    M_z^{j - 1}
    +
    \frac {2j + 1}{2j + 2}
    \left( \frac j{ j + 1 } \right)^{z - 1}
    M_{z - 1}^j
    , \quad
    z = 1, \ldots, j - 1
    , \quad
    j \geqslant 2
    ,
\]
where
\[
    M_0^j = M_j^j = M_j^{j + 1} = 1
    , \quad
    j \in \mathbb N_0
    .
\]
Then it is easy to see that
\[
    \left| \widetilde B_z^j \right|
    \leqslant
    M_z^j
    \sqrt{2j + 1}
\]
holds true for all admissible pairs $z, j$.
These new parameters can be roughly estimated by
\[
    M_z^{j + 2}
    \leqslant
    \max_{z = 0, \ldots, j + 1}
    M_z^{j + 1}
    +
    \max_{z = 0, \ldots, j}
    M_z^j
    \leqslant
    2^j
    , \quad
    z = 1, \ldots, j
    , \quad
    j \geqslant 0
    ,
\]
which leads to the bound for $\widetilde B_m^j$ claimed above at the beginning of the proof.
We finish now by noticing
\begin{equation*}
    \sum_{m=0}^j
    \sum_{q=0}^p
    \frac{ (k + 2 + j + p)! (4j)^m (4p)^q }{ (k + 2 + j + p + m + q)! }
    \leqslant
    \sum_{m=0}^j
    \sum_{q=0}^p
    \frac{ (4j)^m (4p)^q }{ (3 + j + p)^{m + q} }
\end{equation*}
that is a multiplication of two finite geometric series.
Summing these geometric series and using the bound for $\widetilde B_m^j$
we conclude the proof by \eqref{cross_correlation_coefficients}.

\end{proof}

Now we can make a use of the obtained power series for
\(
    \widehat \phi_j (\zeta)
    \widehat \phi_p (-\zeta)
\)
standing in \eqref{sigma_contour_matrix_a}.
The contour $\Gamma$ is parameterized by
\(
    \zeta = \rho e^{- i \pi / 4}
    ,
    \rho \in \mathbb R
    .
\)
One can change the integration variable, and then exchange
integration and summation by appealing to the dominated convergence theorem.
Thus our main operator \eqref{sigma_contour_matrix_a}
has the following expansion
\begin{equation}
\label{operator_correlation_expansion}
    \left[ \sigma_l^n \right]_{pj}
    (a)
    =
    \sum_{k = 0}^{\infty}
    C_{jp}^{2k}
    J_{2k + j + p}(l, a)
    ,
\end{equation}
where $a = t 4^n$
and
\begin{equation}
\label{power_integral}
    J_m
    =
    \frac
    {
        e^{ i \pi (m - 1) / 4 }
    }
    {
        2 \pi ( m + 2 )!
    }
    \int_{\mathbb R}
    \exp
    \left(
        \rho l \exp \left( i \frac \pi 4 \right) - a \rho^2
    \right)
    \rho^m
    d \rho
\end{equation}
satisfying the following relation
\begin{equation}
\label{power_integral_recursion}
    J_{m+1}
    =
    \frac
    {
        i
    }
    {
        2a (m + 3)
    }
    \left(
        l
        J_m
        +
        \frac {m}{(m + 2)}
        J_{m-1}
    \right)
    , \quad
    m = 0, 1, 2, \ldots,
\end{equation}
with $J_{-1} = 0$ and
\begin{equation}
\label{power_integral_0}
    J_0
    =
    \frac{ e^{ -i \frac{\pi}4 } }{ 4 \sqrt{ \pi a } }
    \exp
    \left(
        \frac{il^2}{4a}
    \right)
    .
\end{equation}
The zero-power integral $J_0$ is standard
and the rest are obtained from it via integration by parts.
The expansion \eqref{operator_correlation_expansion} turns out to be very efficient
as long as the scaling order $n$ is big enough for the chosen time step $t$: the smaller the time $t$ the bigger $n$ should be,
in order to avoid that $|J_m|$ become so large that it will affect accuracy due to rounding errors.
For example, $n$ should be bigger than $12$ for $t = 10^{-5}$.

Using this recurrence formula we can prove some very useful
error estimates.
It is worth to point out that
\(
    C_{jp}^k \in \ell^{\infty}( \mathbb N_0 )
    ,
\)
therefore, it makes sense to estimate the sequence of power integrals
in
\(
    \ell^1( \mathbb N_0 )
    .
\)
We note that $J_m \to 0$  as $m \to \infty$, since the series converges.
Appealing repetitively to the recurrence relation we can show that $J_m$ tends to zero faster than any power of $1/m$.
Denoting $b_m = |J_m|$ we can notice that 
%
%
\begin{equation*}
\begin{aligned}
    b_{m+1}
    &\leqslant
    \frac
    {
        |l|
    }
    {
        2a (m + 3)
    }
    b_m
    +
    \frac m{2a(m + 2)(m + 3)} b_{m-1}
    \\
    b_{m+2}
    &\leqslant
    \frac
    {
        |l|
    }
    {
        2a (m + 4)
    }
    b_{m+1}
    +
    \frac {m + 1}{2a(m + 3)(m + 4)} b_m
    \\
    b_{m+3}
    &\leqslant
    \frac
    {
        |l|
    }
    {
        2a (m + 5)
    }
    b_{m+2}
    +
    \frac {m + 2}{2a(m + 4)(m + 5)} b_{m+1}
    \\
    &\ldots \qquad \ldots \qquad \ldots \qquad \ldots \qquad \ldots
\end{aligned}
\end{equation*}
Summing these inequalities one obtains
\begin{multline*}
    b_{m+1} + b_{m+2} + \ldots
    \leqslant
    \frac m{2a(m + 2)(m + 3)} b_{m-1}
    +
    \frac 1{2a(m + 3)}
    \left(
        |l| + 1 - \frac 3{m + 4}
    \right)
    b_m
    \\
    +
    \frac 1{2a(m + 4)}
    \left(
        |l| + 1 - \frac 3{m + 5}
    \right)
    b_{m+1}
    +
    \ldots
\end{multline*}
On the right-hand side, the coefficients standing in front of $b_{m + 1}, b_{m + 2}, \ldots$
form a decreasing sequence
(strictly decreasing for either $m \geqslant 1$ or $l \ne 0$),
and so we can take them out of the brackets.
Moreover, for
\(
    m
    \geqslant
    \frac{|l| + 1}{2a} - \frac 92
    +
    \sqrt{
        \left(
            \frac{|l| + 1}{2a} + \frac 12
        \right) ^2
        -
        \frac 3a
    }
    ,
\)
or any $m$ in case the value under the square root is negative,
we have
\begin{equation*}
    b_{m+1} + b_{m+2} + \ldots
    \leqslant
    \frac m{a(m + 2)(m + 3)} b_{m-1}
    +
    \frac 1{a(m + 3)}
    \left(
        |l| + 1 - \frac 3{m + 4}
    \right)
    b_m
\end{equation*}
that can be used in practice for series termination.
We can even be more specific, since we only need to
sum $b_m$ with either even or odd indices $m$,
as follows
\begin{multline*}
    b_{m+1} + b_{m+3} + \ldots
    \leqslant
    \frac m{2a(m + 2)(m + 3)} b_{m-1}
    +
    \frac {|l|}{2a(m + 3)}
    b_m
    \\
    +
    \frac {m + 2}{2a(m + 4)(m + 5)} b_{m+1}
    +
    \frac
    {
        |l|
    }
    {
        2a (m + 5)
    }
    b_{m+2}
    +
    \ldots
    \leqslant
    \frac m{2a(m + 2)(m + 3)} b_{m-1}
    +
    \frac {|l|}{2a(m + 3)}
    b_m
    \\
    +
    \frac
    {
        |l| + 1
    }
    {
        2a (m + 5)
    }
    \left[
    \frac m{a(m + 2)(m + 3)} b_{m-1}
    +
    \frac 1{a(m + 3)}
    \left(
        |l| + 1 - \frac 3{m + 4}
    \right)
    b_m
    \right]
    .
\end{multline*}
These inequalities may be exploited to
cut the infinite sum in \eqref{operator_correlation_expansion} in practical calculations, once the requested precision is set.
In principle,
we could proceed further and try to obtain an optimal precision-dependent expression to determine how to terminate the series, though for big $|l|$ we anticipate the series cut to be unnecessarily large.
Moreover,
the bound in Lemma \ref{cross_correlation_coefficients_bound_lema}
is in practice very rough,
due to very fast convergence of the integral series.
In most practical calculations it is enough to know that $|C_{jp}^k| < 10$ for $k \leqslant 50$.
\begin{figure}[ht!]
\centering
\begin{subfigure}{0.49\textwidth}
\includegraphics[width=\textwidth]{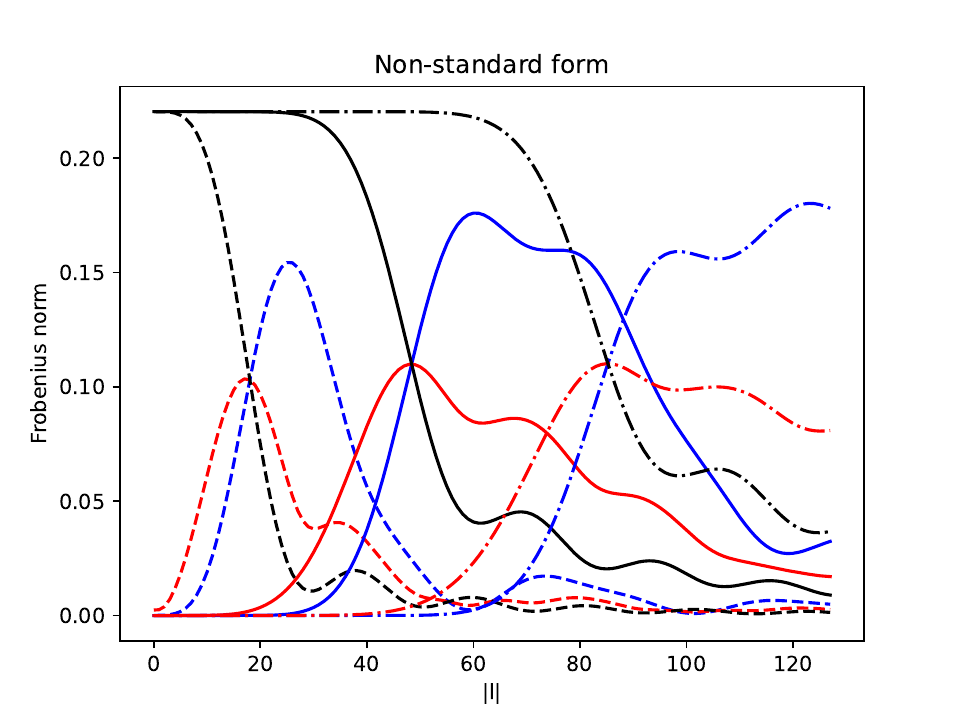}
\caption{$t = 0.0001$, $n=7$}\label{subfig:t0001-n7}
\end{subfigure}
\hfill
\begin{subfigure}{0.49\textwidth}
\includegraphics[width=\textwidth]{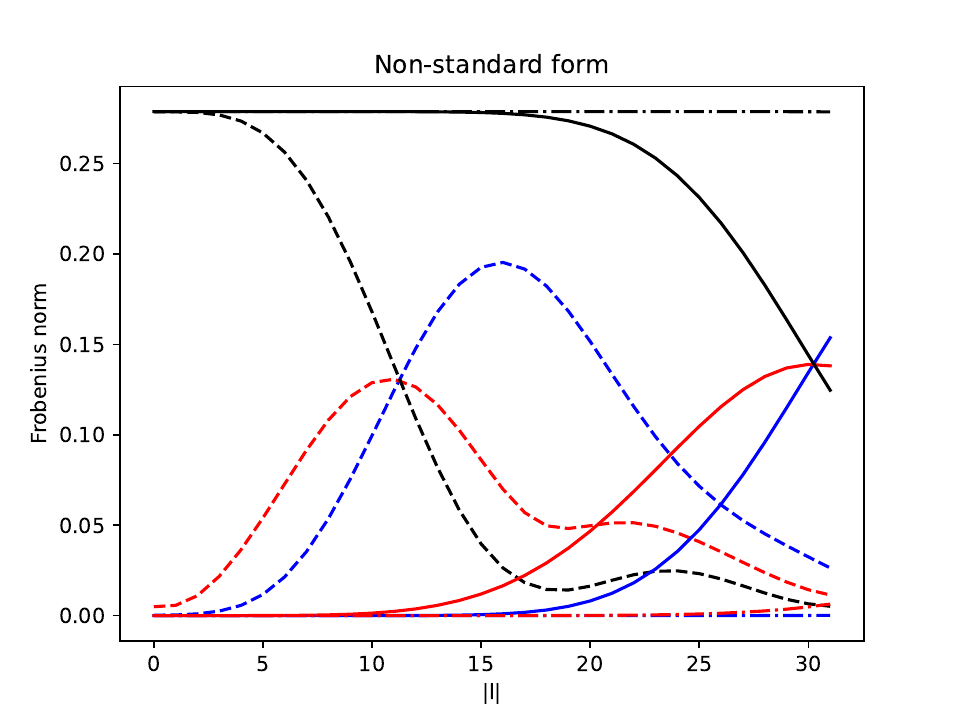}
\caption{$t = 0.001$, $n=5$}\label{subfig:t001-n6}
\end{subfigure}
\hfill
\begin{subfigure}{0.49\textwidth}
\includegraphics[width=\textwidth]{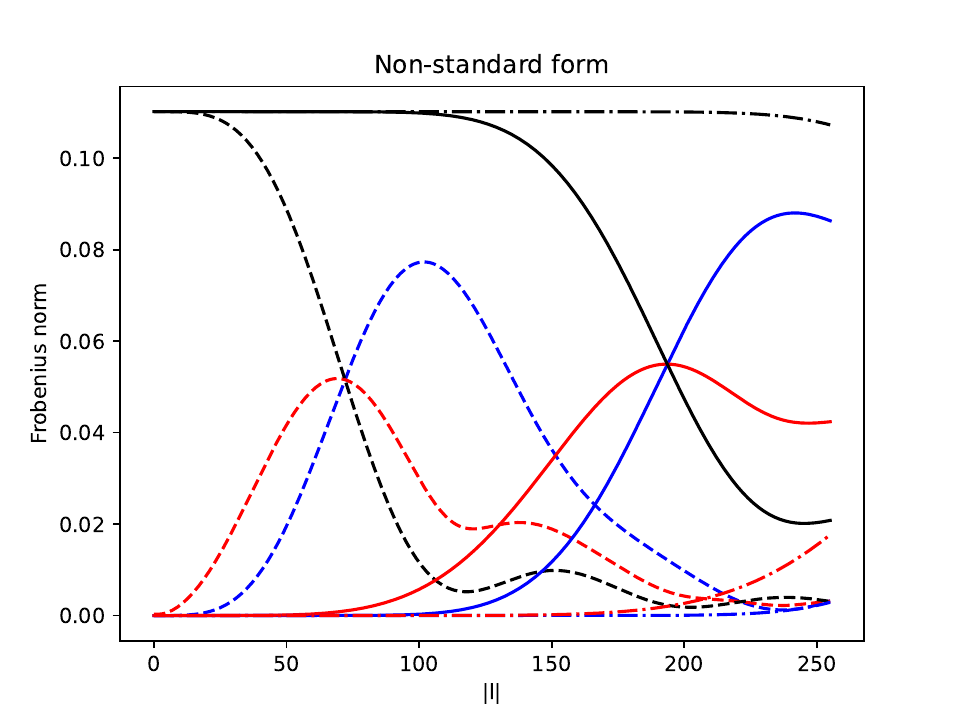}
\caption{$t = 0.0001$, $n=8$}\label{subfig:t0001-n8}
\end{subfigure}
\hfill
\begin{subfigure}{0.49\textwidth}
\includegraphics[width=\textwidth]{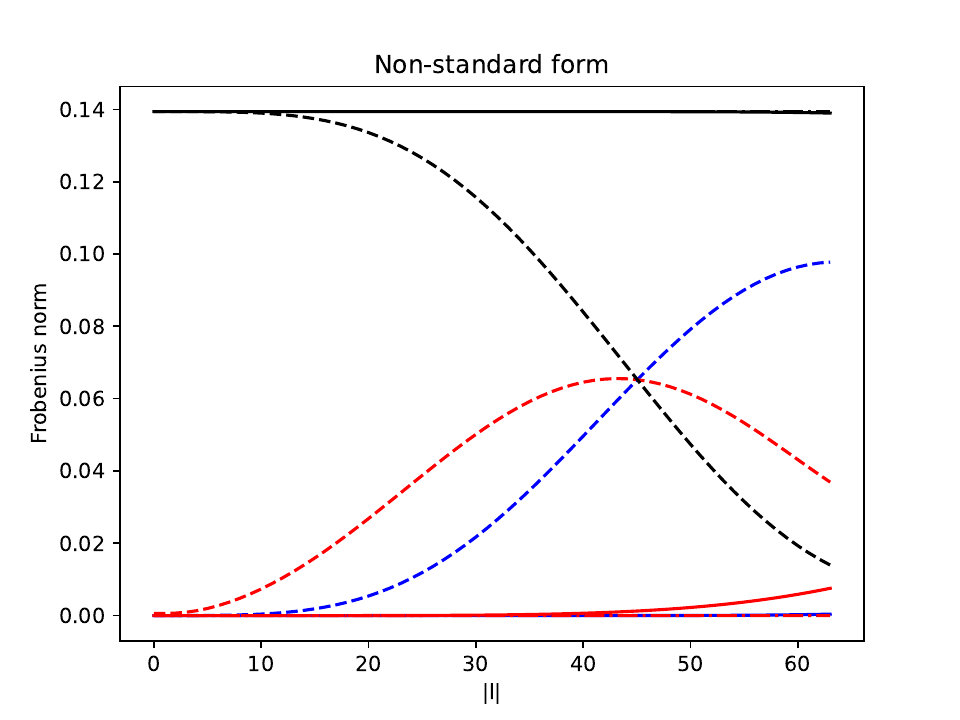}
\caption{$t = 0.001$, $n=6$}\label{subfig:t001-n7}
\end{subfigure}
\hfill
\begin{subfigure}{0.49\textwidth}
\includegraphics[width=\textwidth]{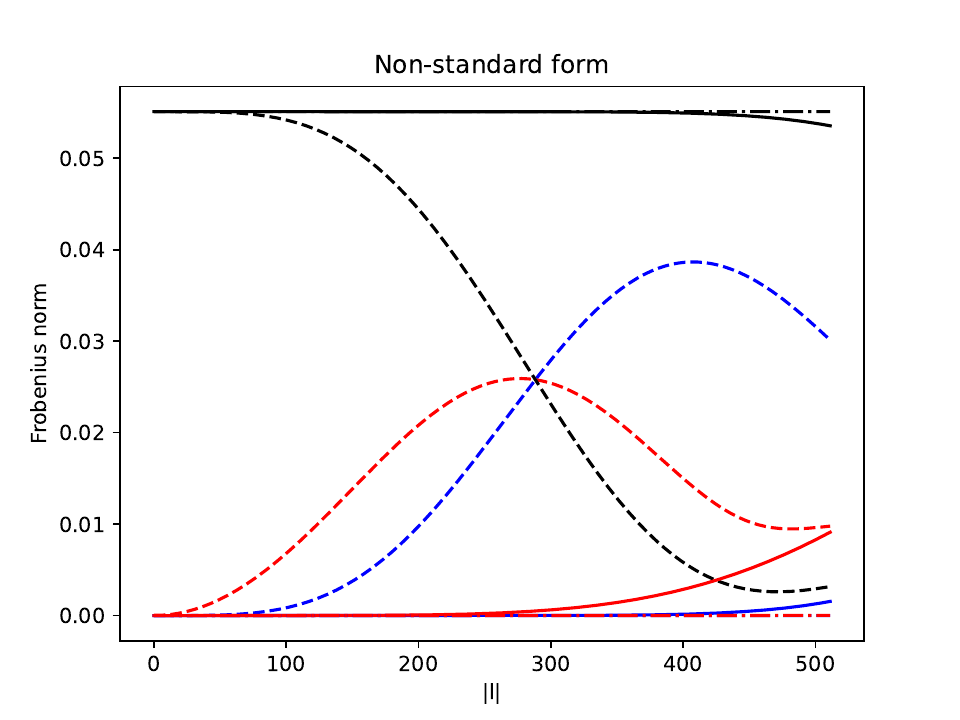}
\caption{$t = 0.0001$, $n=9$}\label{subfig:t0001-n9}
\end{subfigure}
\hfill
\begin{subfigure}{0.49\textwidth}
\includegraphics[width=\textwidth]{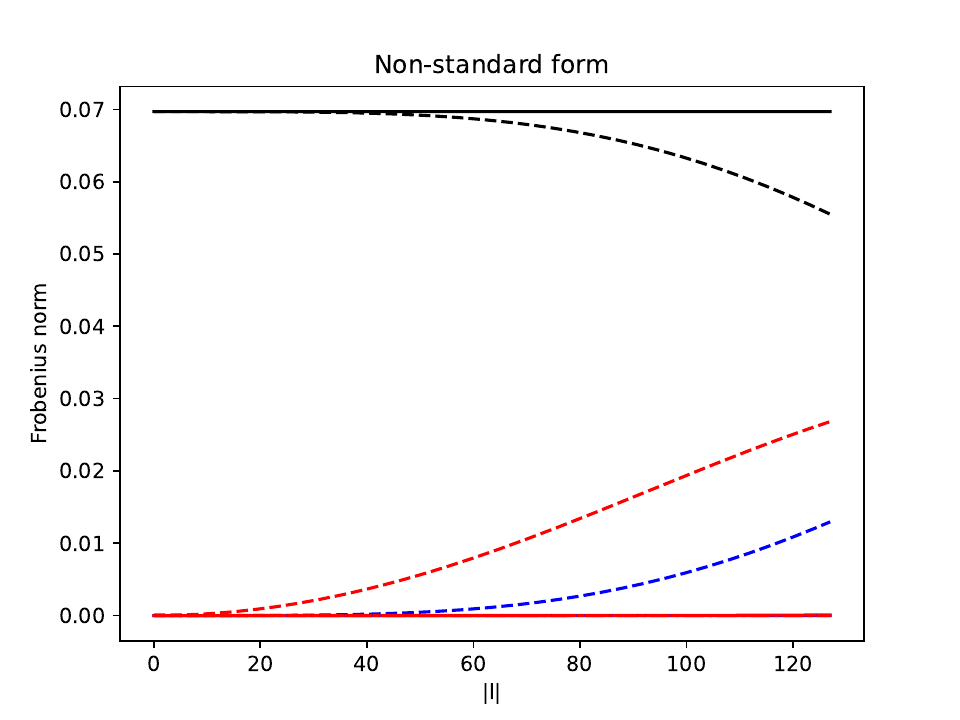}
\caption{$t = 0.001$, $n=7$}\label{subfig:t001-n8}
\end{subfigure}
\hfill
	\caption
	{
		Sparsity of the non-standard form for selected choices of scales and time steps. Left panels show results for time step $t = 0.0001$ and scale from $n=7$ (top) to $n=9$ (bottom). Right panels show results for time step $t = 0.001$ and scale from $n=5$ (top) to $n=7$ (bottom). Each color represent a different component of the \ac{NS}-form: $\norm{ \alpha_l }$ is blue, $\norm{ \beta_l }$ is red and $\norm{ \sigma_l }$ is black. Solid lines correspond to $\mathfrak k = 6$, dashed lines to $\mathfrak k = 2$ and dash-dot lines to $\mathfrak k = 11$.}\label{fig:ns-form-norms}
\end{figure}

\begin{table}[!ht]
    \centering
    \begin{tabular}{ccc|ccc|ccc}
    $t$ & $n$ & $\mathfrak k$ & $\norm{\alpha_0}$ & $\norm{\beta_0}$ & $\norm{\sigma_0}$ & $\norm{ \alpha_{2^n - 1} }$ & $\norm{ \beta_{2^n - 1} }$ & $\norm{ \sigma_{2^n - 1} }$ \\
    \hline
0.0001 & 7 & 2  & 8.0e-05 & 2.4e-03 & 0.22 & 5.0e-03 & 2.6e-03 & 1.3e-3  \\ 
0.0001 & 7 & 6  & 3.2e-13 & 3.9e-08 & 0.22 & 3.2e-02 & 1.7e-02 & 9.0e-3  \\ 
0.0001 & 7 & 11 & 2.7e-18 & 4.0e-15 & 0.22 & 1.8e-01 & 8.1e-02 & 0.037   \\ 
\hline
0.0001 & 8 & 2  & 2.5e-06 & 3.0e-04 & 0.11 & 3.2e-03 & 3.2e-03 & 3.1e-3  \\ 
0.0001 & 8 & 6  & 3.9e-17 & 3.1e-10 & 0.11 & 8.6e-02 & 4.2e-02 & 0.021   \\ 
0.0001 & 8 & 11 & 5.2e-19 & 3.5e-18 & 0.11 & 2.9e-03 & 1.8e-02 & 0.11    \\ 
\hline
0.0001 & 9 & 2  & 7.8e-08 & 3.8e-05 & 0.055 & 3.0e-02 & 9.8e-03 & 3.2e-3 \\ 
0.0001 & 9 & 6  & 5.3e-19 & 2.4e-12 & 0.055 & 1.6e-03 & 9.1e-03 & 0.054  \\ 
0.0001 & 9 & 11 & 3.2e-19 & 1.5e-18 & 0.055 & 6.8e-09 & 1.9e-05 & 0.055  \\      
\hline
\hline
0.001  & 5 & 2  & 2.6e-04 & 4.9e-03 & 0.28 & 2.6e-02 & 1.1e-02 & 5.2     \\ 
0.001  & 5 & 6  & 6.8e-12 & 2.0e-07 & 0.28 & 1.5e-01 & 1.4e-01 & 0.12    \\ 
0.001  & 5 & 11 & 2.4e-18 & 8.5e-14 & 0.28 & 1.5e-04 & 6.4e-03 & 0.28    \\ 
\hline
0.001  & 6 & 2  & 8.1e-06 & 6.1e-04 & 0.14 & 9.8e-02 & 3.7e-02 & 0.014   \\ 
0.001  & 6 & 6  & 8.3e-16 & 1.6e-09 & 0.14 & 4.1e-04 & 7.5e-03 & 0.14    \\ 
0.001  & 6 & 11 & 6.3e-19 & 1.1e-17 & 0.14 & 1.4e-10 & 4.3e-06 & 0.14    \\ 
\hline
0.001  & 7 & 2  & 2.5e-07 & 7.7e-05 & 0.070 & 1.3e-02 & 2.7e-02 & 0.056  \\ 
0.001  & 7 & 6  & 4.2e-19 & 1.2e-11 & 0.070 & 1.1e-07 & 8.8e-05 & 0.070  \\ 
0.001  & 7 & 11 & 2.0e-19 & 1.8e-18 & 0.070 & 3.0e-17 & 1.4e-09 & 0.070  \\ 
    \end{tabular}
    \caption{
        Values of diagonal elements ($|l|=0$) and corner elements ($|l|=2^n-1$)
        for the $\norm{ \alpha_l }$, $\norm{ \beta_l }$ and $\norm{ \sigma_l }$
        components of the \ac{NS} form for selected time steps $t$, scales $n$
        and polynomial order $\mathfrak k$.
    }
    \label{tab:ns-form-norms}
\end{table}

We now turn our attention to the sparsity of the
matrices
\eqref{alpha_contour_matrix}-\eqref{gamma_contour_matrix}
associated with
the nonstandard form \eqref{nonstandard_form}
by means of
\eqref{nonstandard_form_restrictions}, \eqref{nonstandard_form_matrices}.
The matrices
\(
    \sigma_l^{n + 1}
    ,
\) associated with the scaling functions, can be obtained 
using the expansion \eqref{operator_correlation_expansion}, whereas the matrices
\(
    \alpha_l^n, \beta_l^n, \gamma_l^n
\)
can be computed exploiting \eqref{nonstandard_form_convolution_decomposition_step}.
The Frobenius norm
\begin{equation}
\label{Frobenius_norm}
    \norm{ \alpha }
    =
    \left(
        \sum_{p, j = 0}^{\mathfrak k - 1}
        \left|
            \left[ \alpha \right]_{pj}
        \right|^2
    \right)^{1/2}
\end{equation}
is used for a matrix $\alpha$ of $\mathfrak k \times \mathfrak k$-size.
In Figure \ref{fig:ns-form-norms} we have displayed the norm of the different components of the \ac{NS} form as a function of the distance between two nodes $|l|$ for a selection of scales $n$, time steps $t$ and polynomial orders $\mathfrak k$. The values of the norms for $|l| = 0$ (diagonal elements) and $|l| = 2^n-1$ (corner elements) are reported in Table~\ref{tab:ns-form-norms}. These results show that away from the main diagonal $\norm{ \alpha_l^n }$ and $\norm{ \beta_l^n }$ are significant, whereas they become negligible closer to the main diagonal $l = 0$.
This is in contrast with the fact that the \ac{NS}-form matrices of the heat equation
exhibit a narrow-bounded structure along the main diagonal, as mentioned in Subsection~\ref{Convolution_representation_on_the_real_line}.
Intuitively, this difference can be explained as follows.
The heat defuses with time, that is, it transports away from the initial local perturbation point
with some smearing.
Waves described by the Schrödinger equation \eqref{free_particle_schrodinger} disperse.
Such a wave moves away from the initial local perturbation point
forming high frequency ripples that are captured by multiresolution analysis.

Alternatively, the values on the diagonal and away could be roughly estimated as follows.
Due to the vanishing moment property \eqref{wavelet_vanishing_moment},
$\widehat \psi_j(\zeta)$
has a root at zero of order $\mathfrak k + j$.
It results in the fact,
that the corresponding expansions, similar to \eqref{operator_correlation_expansion},
for the matrix elements
\(
    \left[ \alpha_l^n \right]_{pj}
\)
and
\(
    \left[ \beta_l^n \right]_{pj}
\)
start with integrals $J_{2 \mathfrak k + p + j}$
and $J_{\mathfrak k + p + j}$, respectively.
These integrals constitute the leading terms, due to their fast convergence to zero.
Moreover, the corner elements
\(
    \left[ \alpha_l^n \right]_{00}
\)
and
\(
    \left[ \beta_l^n \right]_{00}
\)
are dominant for the same reason.
Therefore,
\(
    \norm{ \alpha_l^n }
    \sim
    | J_{ 2 \mathfrak k }(l) |
\)
and
\(
    \norm{ \beta_l^n }
    \sim
    | J_{ \mathfrak k }(l) |
    .
\)

For $l = 0$ from \eqref{power_integral_recursion} one obtains
\begin{equation}
\label{power_integral_on_diagonal}
    J_{2k}(0)
    =
    \frac
    {
        e^{ i \pi (2k - 1) / 4 }
    }
    {
        2 \sqrt{\pi a} (2a)^k ( 2k + 2 )!! ( 2k + 1 )
    }
    , \quad
    J_{2k + 1}(0) = 0
    .
\end{equation}
In the case of $|l| \gg 1$,
we can admit the following approximation
\begin{equation*}
    J_{m+1}(l)
    \approx
    \frac
    {
        il
    }
    {
        2a (m + 3)
    }
    J_m(l)
\end{equation*}
and so
\begin{equation}
\label{power_integral_outside_diagonal}
    J_m(l)
    \approx
    \left(
        \frac {il}{2a}
    \right)^m
    \frac
    {
        \exp
        \left(
            \frac{il^2}{4a}
            -
            \frac{i\pi}4
        \right)
    }
    {
        2 \sqrt{\pi a}
        (m + 2)!
    }
\end{equation}
which leads to
\[
    \norm{ \alpha_l^n }
    \lesssim
    \frac 1{ 2 \sqrt{\pi a} (2 \mathfrak k + 2)! }
    \left(
        \frac l{2a}
    \right)^{ 2 \mathfrak k }
    , \quad
    |l| \gg 1
\]
Meanwhile,
\(
    \norm{ \alpha_0^n }
    \sim
    | J_{ 2 \mathfrak k }(0) |
    ,
\)
where the integral is calculated in \eqref{power_integral_on_diagonal}.
Thus
\[
    \norm{ \alpha_l^n }
    \lesssim
    \frac
    {
        1
    }
    {
        2 \sqrt{\pi a} (2a)^{\mathfrak k}
        ( 2 \mathfrak k + 2 )!! ( 2 \mathfrak k + 1 )
    }
    +
    \frac 1{ 2 \sqrt{\pi a} (2 \mathfrak k + 2)! }
    \left(
        \frac l{2a}
    \right)^{ 2 \mathfrak k }
\]
Similarly,
\[
    \norm{ \beta_l^n }
    =
    \norm{ \gamma_l^n }
    \lesssim
    \frac
    {
        1
    }
    {
        2 \sqrt{\pi a} (2a)^{\mathfrak k / 2}
        ( \mathfrak k + 2 )!! ( \mathfrak k + 1 )
    }
    +
    \frac 1{ 2 \sqrt{\pi a} (\mathfrak k + 2)! }
    \left(
        \frac l{2a}
    \right)^{ \mathfrak k }
    , \quad
    \text{provided $\mathfrak k$ is even}
\]
and
\[
    \norm{ \beta_l^n }
    =
    \norm{ \gamma_l^n }
    \lesssim
    \frac
    {
        1
    }
    {
        2 \sqrt{\pi a} (2a)^{(\mathfrak k + 1) / 2}
        ( \mathfrak k + 3 )!! ( \mathfrak k + 2 )
    }
    +
    \frac 1{ 2 \sqrt{\pi a} (\mathfrak k + 2)! }
    \left(
        \frac l{2a}
    \right)^{ \mathfrak k }
    , \quad
    \text{provided $\mathfrak k$ is odd.}
\]
Note that the expansions for matrices $\alpha_l^n$, $\beta_l^n$
with respect to the power integrals contain cross correlation coefficients
depending on the MRA order $\mathfrak k$.
Therefore, the implicit constants staying in the obtained inequalities for
\(
    \norm{ \alpha_l^n }
    ,
    \norm{ \beta_l^n }
\)
depend on $\mathfrak k$ as well.
In other words,
these inequalities provide only a qualitative behaviour of the norms.
For instance,
one can compare
\[
    | J_{ 2 \mathfrak k }(0) |
    =
    \left \{
        \begin{aligned}
            \text{8.6e-05, } \quad \mathfrak k = 2
            \\
            \text{2.1e-11, } \quad \mathfrak k = 6
            \\
            \text{1.0e-20,} \quad \mathfrak k = 11
        \end{aligned}
    \right.
\]
for $n = 7$ and $t = 0.0001$ with the values
\(
    \norm{ \alpha_0^n }
\)
reported in Table~\ref{tab:ns-form-norms}.

\subsection{Interpolating scaling functions}\label{sec:interpolating}

Another common choice of polynomial basis is constituted by the interpolating scaling functions.
Similarly to the previous case
we introduce the functions
\(
    \Psi_j(-i \zeta)
    =
    \widehat \varphi_j(\zeta)
    ,
\)
where $\varphi_j$ are given by \eqref{interpolating_scaling_functions},
so that
\[
    \Psi_j(x)
    =
    \sqrt{ w_j }
    \sum_{m = 0}^{\mathfrak k - 1}
    \phi_m(x_j) \Phi_m(x)
    , \quad
    j = 0, \ldots, \mathfrak k - 1
    ,
\]
where $\phi_m$ are the Legendre scaling functions.
Here $x_0, \ldots, x_{\mathfrak k - 1}$ denote the roots of $P_{\mathfrak k}(2x - 1)$.
We are interested in the following combination
\begin{equation*}
    \Psi_j (-x)
    \Psi_p (x)
    =
    \sqrt{ w_j w_p }
    \sum_{j', p' = 0}^{\mathfrak k - 1}
    \phi_{j'}(x_j) \phi_{p'}(x_p) \Phi_{j'}(-x) \Phi_{p'}(x)
\end{equation*}
which can be expressed as the series
\begin{equation*}
    \Psi_j (-x)
    \Psi_p (x)
    =
    \sum_{k = 0}^{\infty}
    \frac{ D_{jp}^k x^k }{ (k + 2)! }
    ,
\end{equation*}
where
\begin{equation*}
    D_{jp}^k
    =
    \sqrt{ w_j w_p }
    \sum_{
        \substack
        {
            0 \leqslant j', p' \leqslant \mathfrak k - 1
            \\
            j' + p' \leqslant k
        }
    }
    \phi_{j'}(x_j) \phi_{p'}(x_p) C_{j'p'}^{k - j' - p'}
    , \quad
    j, p = 0, \ldots, \mathfrak k - 1
    .
\end{equation*}
Note that this double sum can be restricted to
\(
    k - j' - p'
\)
being even and non-negative.
Thus the time evolution operator in the interpolating basis has the following expansion
\begin{equation}
    \left[ \sigma_l^n \right]_{pj}
    (a)
    =
    \sum_{k = 0}^{\infty}
    D_{jp}^k
    J_k(l, a)
    ,
\end{equation}
where $a = t 4^n$ as above for the Legendre basis.

Calculating matrices
\(
    \alpha_l^n, \beta_l^n, \gamma_l^n
\)
from
\(
    \sigma_l^{n + 1}
\)
by \eqref{nonstandard_form_convolution_decomposition_step}
and then evaluating the norms
\(
    \norm{\alpha_l^n}
    ,
    \norm{\beta_l^n}
\)
and
\(
    \norm{\gamma_l^n}
    ,
\)
one arrives to the same qualitative results as in the previous subsection.
Moreover, the interpolating basis spans the same space
and the norms of the operator therefore do not change,
and as a result matrix norms may change insignificantly.
Therefore, we omit the illustration of the sparsity and refer to Figure~\ref{fig:ns-form-norms} for details.

\subsection{Haar multiresolution analysis}
\label{Haar_multiresolution_analysis_subsection}

The above conclusions simplify significantly in the case $\mathfrak k = 1$.
By \eqref{sigma_contour_matrix_a}, \eqref{Fourier_Legendre_0} we have
\begin{equation*}
    \left[ \sigma_l^{n} \right]_{00}
    =
    \int_{\Gamma}
    F (\zeta)
    d\zeta
    , \quad
    F(\zeta)
    =
    \frac{  1 - \cos \zeta }{ \pi \zeta^2 }
    \exp
    \left(
        i \zeta l - it N^2 \zeta^2
    \right)
    .
\end{equation*}
Let us first justify \eqref{H_bound},
so we regard $\zeta$ lying in the band $| \im \zeta | \leqslant H$
and satisfying $\re \zeta \cdot \im \zeta \leqslant 0$,
then
\[
    \Abs{F(\zeta)}
    \leqslant
    \frac 2{\pi}
    e^{H(|l| + 1)}
    \leqslant
    \frac 2{\pi}
    e^{HN}
    .
\]
In order to maintain accuracy $\varepsilon$, while discretizing the integral in such domain,
the right-hand side of this bound should not be bigger than $\varepsilon / \varepsilon_{\text{mach}}$.
Thus we obtained \eqref{H_bound}.

Now we have only even-power integrals in the expansion \eqref{operator_correlation_expansion},
\begin{equation}
\label{haar_sigma}
    \left[ \sigma_l^{n} \right]_{00}
    =
    2
    \sum_{k = 0}^{\infty}
    J_{2k}(l, a)
    ,
\end{equation}
since $C_{00}^k = 1 + (-1)^k$ by \eqref{cross_correlation_coefficients},
that can also be checked directly by expanding $\cos \zeta$ about zero.
The convergence rate of this series was analysed
in Subsection~\ref{Legendre_scaling_functions_subsection}.

Finally, let us have a closer look at the dependence of $\left[ \sigma_l^{n} \right]_{00}$
on the distance $|l|$ to the diagonal.
In the case of $l = 0$, one can calculate the sum of the series
by means of special functions as
\begin{equation}
\label{haar_sigma_on_diagonal}
    \left[ \sigma_0^{n} \right]_{00}
    =
    - i \sqrt{ \frac 2{\pi} }
    F_1
    \left(
        \frac{ e^{ \frac{i\pi}4 } }{ \sqrt{2a} }
    \right)
    ,
\end{equation}
where
\begin{equation*}
    F_1(x)
    =
    \sum_{k = 0}^{\infty}
    \frac
    {
        x^{ 2k + 1 }
    }
    {
        ( 2k + 1 )( 2k + 2 )!!
    }
    =
    \frac{1 - e^{x^2/2}}x
    -
    i \sqrt{ \frac {\pi}2 } \erf \left( \frac{ix}{\sqrt 2} \right)
    ,
\end{equation*}
using \eqref{power_integral_on_diagonal}.
The series
\(
    F_1
    \left(
        \frac{ e^{ \frac{i\pi}4 } }{ \sqrt{2a} }
    \right)
\)
is very useful in practice for the precision control,
since its real and imaginary
parts are alternating series with monotonically decreasing coefficients.

In the case of $|l| \gg 1$,
we have \eqref{power_integral_outside_diagonal}
which leads to
\begin{equation}
\label{haar_sigma_outside_diagonal}
    \left[ \sigma_l^{n} \right]_{00}
    \approx
    \frac 1{ \sqrt{ \pi a } }
    \left(
        \frac {2a}l
    \right)^2
    \left(
        1 - \cos \frac l{2a}
    \right)
    \exp
    \left(
        \frac{il^2}{4a} - \frac{i\pi}4
    \right)
    .
\end{equation}
Here the term in front of the exponent changes slowly with $l$,
whereas the exponent itself oscillates very fast.
This formula describes the operator matrix with a very high precision
for any $l \ne 0$.
In fact, the limit of this expression as $l/{2a} \to 0$ is identical to
$\left[ \sigma_0^{n} \right]_{00}$ evaluated by \eqref{haar_sigma_on_diagonal}
with the first order approximation $F_1(x) \approx x / 2$.
Moreover,
\eqref{haar_sigma_on_diagonal} and \eqref{haar_sigma_outside_diagonal}
give us an idea of how many terms are important in Expansion \eqref{haar_sigma}
for a given precision.
The same series cut is reliable in practise for evaluation of
the general expansion \eqref{operator_correlation_expansion}.

\section*{Declarations}

The CRediT taxonomy of contributor roles \cite{Allen2014-rd,Brand2015-qc} is applied.
The ``Investigation'' role also includes the ``Methodology'', ``Software'', and ``Validation'' roles.
The ``Analysis'' role also includes the ``Formal analysis'' and ``Visualization'' roles.
The ``Funding acquisition'' role also includes the ``Resources'' role.
The contributor roles are visualized in the following authorship attribution matrix,
as suggested in Ref.~\cite{authorship-attribution}.

\begin{table}[ht]
\caption{Levels of contribution: \textcolor{blue!100}{major}, \textcolor{blue!25}{minor}.}
\label{tbl:contribs}
\begin{tabular}{lccc}
\hline\hline
                              & ED     & YZ     & LF     \\ \hline
    Conceptualization         & \major &        & \minor \\ 
    Investigation             & \major &        &        \\ 
    Data curation             & \major &        &        \\ 
    Analysis                  & \major & \minor &        \\ 
    Supervision               & \major &        &        \\ 
    Writing -- original draft & \major &        &        \\
    Writing -- revisions      & \major &        & \minor \\
    Funding acquisition       & \minor &        & \major \\ 
    Project administration    & \major &        &        \\ \hline\hline
\end{tabular}
\end{table}



\vskip 0.05in
\noindent
{\bf Acknowledgments.}
{
    The authors are grateful to
    S. R. Jensen and Ch. Tantardini for numerous helpful discussions.
We acknowledge support from the Research Council of Norway through its Centres of Excellence scheme (Hylleraas centre, 262695), through the FRIPRO grant ReMRChem (324590),  and from NOTUR -- The Norwegian Metacenter for Computational Science through grant of computer time (nn14654k).
}

\bibliographystyle{acm}
\bibliography{bibliography}

\end{document}